\documentclass[10pt,reqno]{amsart}
\usepackage{hyperref}
\usepackage{float}
\usepackage{mathtools}
\usepackage[all]{xy}

\usepackage{tikz-cd}
\usepackage{quiver}
\usepackage{citeref}
\usepackage{amscd,amssymb,amsmath,latexsym,bm}
\usepackage[mathcal,mathscr]{euscript}
\usepackage{lipsum}
\usepackage{amsfonts}
\usepackage{graphicx}
\usepackage{epstopdf}
\usepackage{algorithmic}
\usepackage{hyperref}
\usepackage{xcolor}
\usepackage{upgreek}
\usepackage{enumerate}
\usepackage{ulem}
\usepackage{stmaryrd}
\usepackage{gensymb}			%DJA added this for degrees symbol
\usepackage[utf8]{inputenc}
 
\newtheorem{theorem}{Theorem}[section]

\newtheorem{corollary}[theorem]{Corollary}
\newtheorem{proposition}[theorem]{Proposition}
\newtheorem{remark}[theorem]{Remark}
\newtheorem{definition}[theorem]{Definition}
\newtheorem{example}[theorem]{Example}

\numberwithin{equation}{section}
\numberwithin{figure}{section}
\newcommand{\BM}{{\mathbb B}}
\newcommand{\CM}{{\mathbb C}}
\newcommand{\NM}{{\mathbb N}}

\newcommand{\RM}{{\mathbb R}}

\newcommand{\ZM}{{\mathbb Z}}

\newcommand{\FM}{{\mathbb F}}

\newcommand{\Aa}{{\mathcal A}}

\newcommand{\Bb}{{\mathcal B}}

\newcommand{\Ff}{{\mathcal F}}
\newcommand{\Gg}{{\mathcal G}}

\newcommand{\Ss}{{\mathcal S}}

\newcommand{\Tt}{{\mathcal T}}
\newcommand{\Rr}{{\mathcal R}}

\newcommand{\Cc}{{\mathcal C}}

\newcommand{\Hh}{{\mathcal H}}

\begin{document}
\title[Cayley-Crystals]{Spectral and Combinatorial Aspects of Cayley-Crystals}

\author{Fabian R. Lux}

\address{Department of Physics and
\\ Department of Mathematical Sciences 
\\Yeshiva University 
\\New York, NY 10016, USA \\
\href{mailto:fabian.lux@yu.edu}{fabian.lux@yu.edu}}

\author{Emil Prodan}

\address{Department of Physics and
\\ Department of Mathematical Sciences 
\\Yeshiva University 
\\New York, NY 10016, USA \\
\href{mailto:prodan@yu.edu}{prodan@yu.edu}}

\date{\today}

\begin{abstract} 
Owing to their interesting spectral properties, the synthetic crystals over lattices other than regular Euclidean lattices, such as hyperbolic and fractal ones, have attracted renewed attention, especially from materials and meta-materials research communities. They can be studied under the umbrella of quantum dynamics over Cayley graphs of finitely generated groups. In this work, we investigate numerical aspects related to the quantum dynamics over such Cayley graphs. Using an algebraic formulation of the ``periodic boundary condition'' due to L\"uck [Geom. Funct. Anal. {\bf 4}, 455–481 (1994)], we devise a practical and converging numerical method that resolves the true bulk spectrum of the Hamiltonians. Exact results on the matrix elements of the resolvent, derived from the combinatorics of the Cayley graphs, give us the means to validate our algorithms and also to obtain new combinatorial statements. Our results open the systematic research of quantum dynamics over Cayley graphs of a very large family of finitely generated groups, which includes the free and Fuchsian groups.
\end{abstract}

\thanks{This work was supported by the U.S. National Science Foundation through the grants DMR-1823800 and CMMI-2131760.}

\maketitle

{\scriptsize \tableofcontents}

\setcounter{tocdepth}{1}

\section{Introduction and Main Statements}
\label{Sec:Introduction}

Given a discrete group $G$ and a finite subset $S \subset G$, its associated Cayley digraph $\Cc(G,S)$ is the colored graph with vertex set $G$ and directed edges from $g$ to $sg$ for $g \in G$ and $s \in S$, with one distinct color for each edge produced by $s \in S$. If $S$ is the set of generators in a standard presentation of $G$, then we are talking about the standard Cayley diagraph of $G$, which we denote by $\Cc(G)$. The standard Cayley graph can be often represented as a geometric graph rendered in our physical space and, when this is possible, it encodes the entire algebraic information on $G$ in a geometric fashion. Quite often, difficult algebraic and analytic problems can be solved using the geometries of the Cayley graphs \cite{MeierBook}. For example, in this work, we examine the close relation between the spectral characteristics of the quantum dynamics over a Cayley graph and the combinatorics of the graph. In the same time, Cayley graphs are an abundant source of interesting lattices that can be used to systematically explore the world of discrete patterns \cite{KollarCMP2020,BoettcherPRB2022}. In materials science, this is interesting because the current technologies are at a point where fabrications of such lattices are feasible at many different length scales. As such, we can observe and learn from the synthetic quantum or classical dynamics set in play by the degrees of freedom of these new physical systems. In practice, this amounts to placing quantum resonators at the nodes or edges of the lattices and coupling these resonators in a fashion that respects the symmetries of the graphs. For example, in a natural crystal, the quantum resonators are the atoms themselves and the crystal's sites are determined by a decorated Cayley graph generated from an Euclidean space group. Synthetic hyperbolic crystals can and have been fabricated by rendering superconducting circuits over a Cayley graph generated from a Fuchsian group \cite{KollarNature2019}. Classical resonators, rendered and weakly coupled over a Cayley graph, generate a dynamics that is akin to a quantum dynamics (see \cite{LenggenhagerNatComm2022,ZhangNatComm2022,RuzzeneEM2021,ChengPRL2022} for examples relevant to the present context). It was also shown \cite{ProdanArxiv2022} that any quantum dynamics over a Cayley graph can be reproduced with stochastic dynamics over a decorated version of the same graph.

The Hilbert space for quantum dynamics over a Cayley graph of a discrete group $G$ is $\ell^2(G)$. It accepts the orthonormal basis $|g\rangle$, $g \in G$, and the group $G$ acts on $\ell^2(G)$ via either the right regular representation, $\pi_R(g) |g'\rangle = |g' g^{-1}\rangle$, or the left regular representation $\pi_L(g) |g'\rangle = |gg'\rangle$. The generic Hamiltonians take the form $H = \sum_{g,g' \in G} w_{g',g} \, |g' \rangle \langle g |$ and self-adjointness imposes the constraint $w_{g,g'} = w_{g',g}^\ast$. Furthermore, if the quantum dynamics is invariant against the right action of the group $G$, the coefficients of the  Hamiltonian must display the additional constraint $w_{g'h,gh} = w_{g',g}$ for any $g$, $g'$ and $h$ from $G$. As we shall see in section~\ref{Sec:QDynamics}, any such $H$ can be generated as the left regular representation of the element $\sum_{q \in G} w_{q,1} \cdot q$ from the group algebra $\CM G$. This is how the group algebra enters the picture. The case $G =\ZM$, the group of integers, supplies simple examples. Indeed, the Cayley graph of $\ZM$ is the regular 1-dimensional lattice and a translational invariant Hamiltonian takes the form 
\begin{equation}\label{Eq:HZ}
H = \sum\nolimits_{k \in \ZM} \sum\nolimits_{n \in \ZM} \big (t_k |n+k\rangle \langle n | + t^\ast_k |n\rangle \langle n+k|\big),
\end{equation}
where the sum over $k$ is finite. Using the shift operator $T = \sum_{n\in \ZM} |n+1\rangle \langle n|$, this Hamiltonian can be conveniently written as $H = \sum\nolimits_{k \in \ZM}\big ( t_k T^{k} + t^\ast_k T^{-k} \big )$ and, since $T|n\rangle =|n+1\rangle$, we see that $T$ is just the left representation of 1, the generator of $\ZM$. We will use this example below to demonstrate how periodic boundary conditions can be formulated in an algebraic manner.

In the era of topological quantum materials, the most interesting feature of a quantum dynamics is the structure of the gaps in its resonant spectrum, which will be referred to as bulk spectral gaps if the sample is infinite. Indeed, with or without synthetic gauge fields, these gaps can carry topological invariants, typically detected by the $K$-theory of the underlying algebra of physical observables \cite{Bellissard1986,Bellissard1995,KellendonkRMP95}. Nontrivial topological invariants set in action interesting physics in the presence of lattice defects \cite{ProdanJPA2021}. Therefore, one of the important tasks when studying the quantum dynamics over a graph is resolving the gaps in its bulk spectrum. In a laboratory or in a computer simulation, one can only deal with a finite graph and, without proper boundary conditions, the resonant spectrum is contaminated by boundary modes and the detection of bulk spectral gaps may become impossible. This is particularly acute for the free and Fuchsian groups, as well as for any other non-amenable group, for which the ratio between the sites located at the boundary and the total number of sites of a truncated Cayley graph does not decrease to zero with the size of the truncated graph (see the estimate in Remark~\ref{Re:Ratio}). In such cases, it is quite obvious that specialized boundary conditions are needed in order to converge the spectrum or the spectral density of a finite quantum model to its thermodynamic limit, and this is the subject of our paper.

\begin{figure}[t]
\center
\includegraphics[width=0.5\textwidth]{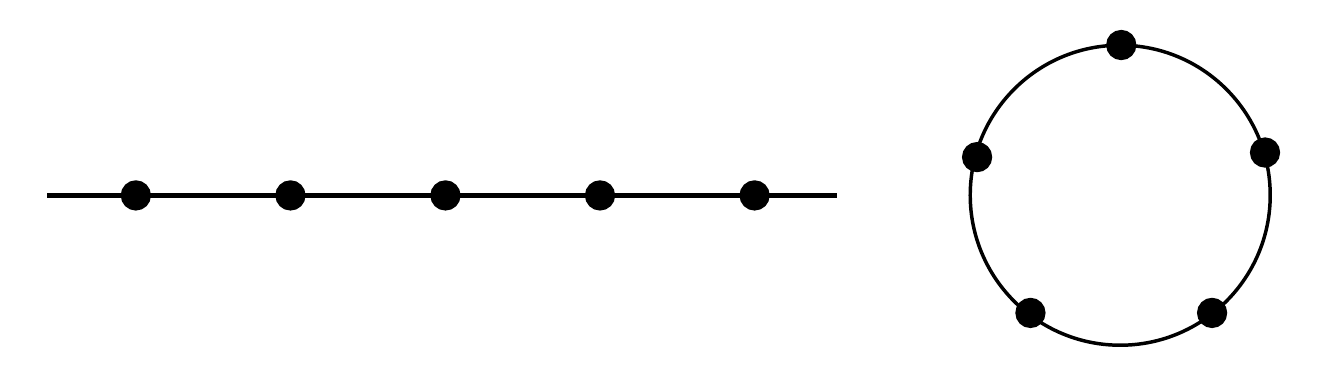}\\
  \caption{\small Left: After truncation of an infinite chain, dangling bonds are produced at the edges. Right: The dangling bonds are removed by forming a bond between the edges.
}
 \label{Fig:DB1}
\end{figure}

For a quick orientation, let us recall first the context of an infinite 1-dimensional chain. Here, a truncation produces dangling bonds, depicted in Fig.~\ref{Fig:DB1} as the broken connections at the ends of the finite chain. Typically, these dangling bonds cause a contamination of the spectrum, which disappears once they are removed by closing the finite chain in the geometry shown in the right panel of Fig.~\ref{Fig:DB1}. Furthermore, with such closed geometries, the thermodynamic limit of an intensive variables is achieved exponentially fast with the size of the truncation. This is a heuristic principle, well known to the physicists and chemists, but the textbooks rarely explain the real mechanism behind it. Now, let us try to remove the dangling bonds resulted from a truncation of a Cayley graph of a non amenable group, such as the free non-abelian group $\FM_2$ with two generators. A section of its Cayley graph is shown in Fig.~\ref{Fig:DB2}, where the reader can see that a truncation that includes all words of length up to three produces a large number of dangling bonds. This number increases exponentially with the size of the truncation. Now, as explained in section~\ref{Sec:CGraphs}, the actions of the two generators of $\FM_2$ produce the flows on the Cayley graph shown in Fig.~\ref{Fig:DB2} by the blue and green arrows. The orbits of the flows produce distinct oriented paths, as the ones highlighted in Fig.~\ref{Fig:DB2}, which can be closed as shown. Note that this procedure can be repeated for any discrete group and, in the case of $\ZM^d$, it generates the ordinary periodic boundary conditions. Now, in the case of $\FM_2$, while this procedure leaves no dangling bonds behind and it can be repeated for increasing sizes of truncations, a calculation of the spectral density function for the graph's adjacency operator, presented in section~\ref{Sec:Numerics}, shows that it does not converge to the correct result with the size of the truncation. In contradistinction, with the specialized periodic boundary conditions derived in subsection~\ref{Sec:ConvF2}, the spectral density function rapidly converges to the correct result. The removal of the dangling bonds by this specialized periodic boundary condition is very complex and would be difficult to guess without the tools developed in this work. Despite the mentioned complexity, we want to assure the reader that our procedure is algorithmic and it can be efficiently implemented on a computer. Such computer codes are supplied in \cite{FabianGH}, together with proper documentation.

\begin{figure}[t]
\center
\includegraphics[width=0.3\textwidth]{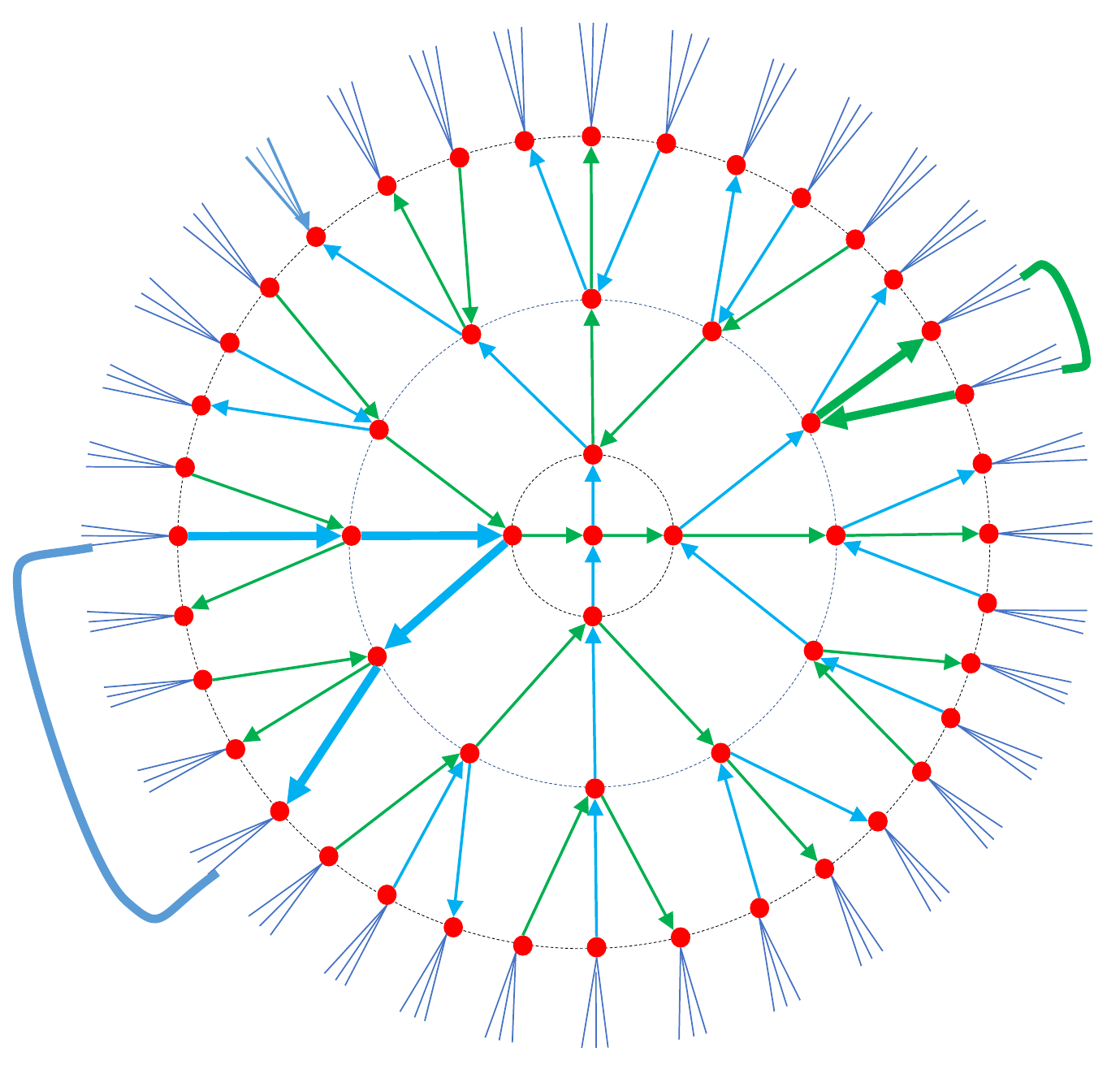}\\
  \caption{\small Truncation of the Cayley graph of $\FM_2$ produces a large number of dangling bonds. They can be removed by closing the orbits generated by the actions of the generators on the graph.
}
 \label{Fig:DB2}
\end{figure}

Even before approaching this problem of convergence, there is the fundamental question about the nature of the bulk spectrum. Indeed, while the generators of the quantum dynamics over a Cayley graph are drawn from the group algebra $\CM G$, in order to do spectral theory and functional calculus, this algebra needs to be completed to a $C^\ast$-algebra. For non-amenable groups, there are two distinct options, the full and the reduced group $C^\ast$-algebras \cite[Sec.~VII2]{DavidsonBook}. We recall that, for a non-amenable group, a functorial relation exists only between the group and its full $C^\ast$-algebra. For example, the irreducible representations of the group are in one-to-one relation with the representations of the full group $C^\ast$-algebra \cite[p.~184]{DavidsonBook}, while the representations of the reduced group $C^\ast$-algebra are generally very poor. As such, the spectrum of an element from $\CM G$ depends on which of these two algebras is used (see Example~\ref{Ex:SpurSpec}). The dynamics on $\ell^2(G)$ produced by the Hamiltonians mentioned in second paragraph, and observed in the typical experiments, takes place in the reduced group $C^\ast$-algebra, as canonically  embedded as a sub-algebra of $\BM(\ell^2(G))$, the algebra of bounded operators on $\ell^2(G)$ (see subsection~\ref{Sec:GCAlg}). Indeed, any operator produced by a continuous functional calculus with $H$ stays in the reduced group $C^\ast$-algebra. In particular, this is the case for the resolvent $(H-z)^{-1}$ which determines the spectrum of $H$. Hence, there is an additional constraint on the boundary conditions, in that they not only need to  obstruct the boundary modes from appearing, but they also have to ensure that the finite models converge to infinite models drawn from the reduced $C^\ast$-algebra of the group. In our section~\ref{Sec:Numerics}, the reader can find numerical examples where different boundary conditions drive the calculations into both mentioned $C^\ast$-algebras. 

The above issues can be attacked from the combinatorial side and, for the particular case of free groups, McKay has provided in \cite{McKayLG1981} an explicit criterion that ensures the correct spectral convergence for a sequence of finite graph approximations. The criterion is actually very simple: The finite graph approximations have to be regular and the ratio between their number of closed loops and their number of vertices needs to converge to zero. The effectiveness of this criterion was put to the test in \cite{SchumacherAHP2015} and the numerical evidence there is promising. Unfortunately, to our knowledge, this criterion does not afford a generalization to the case when group relations are present and finite closed loops appear in the graph. For these cases, we turned to the algebraic solution introduced by L\"uck in \cite{LuekGFA1994} (see also \cite{LuekBook2002} for a fascinating account of applications and \cite{DodziukJFA1998,FarberMA1998,SchickMA2000,SchickTAMS2001,DodziukCPAM2003} for additional approximation results). There, it is shown that, if a group $G$ accepts a coherent sequence of finite index normal subgroups $G =G_0 \triangleright G_1 \triangleright G_2 \cdots$ such that $\bigcap G_n =\{e\}$, the trivial group, then the diagonal matrix elements of the Borel functional calculus in the group von Neumann algebra of $G$ for a self-adjoint element $h \in \CM G$ can be retrieved from Borel functional calculi in the finite von Neumann algebras associated to the quotients $G/G_n$. In Corollary~\ref{Cor:FuncCalc}, we produce a statement about the off-diagonal matrix elements, which in plain words can be stated as follows: Using the standard projections $\phi_n : G \to G/G_n$, we can generate a sequence of approximate Hamiltonians $\phi_n(h) \in \CM G/G_n$, where $h \in  \CM G$ is the element that gives the Hamiltonian $H$ via the left regular representation. Let $H_n$ be the left regular representation of $\phi_n(h)$ on the {\it finite} Hilbert space $\ell^2(G/G_n)$. Then 
\begin{equation}\label{Eq:S0}
 \langle g | (H-z)^{-1} |g'\rangle = \lim_{n \to \infty} \langle \phi_n(g) |(H_n - z)^{-1}|\phi_n(g') \rangle, \quad Im(z) \neq 0.
\end{equation}

A group $G$ for which such coherent sequences of normal subgroups exist is called a residually finite group. Any finitely generated group possessing a faithful representation into $GL(n, F)$ for $F$ a field is residually finite \cite[Th.~4.2]{WehrfritzBook}. In particular, the free groups and Fuchsian groups are residually finite. After studying the arguments leading to this conclusion in \cite{WehrfritzBook}, we found a simple mechanism to systematically generate coherent sequences of normal subgroups: If $R$ is the finitely generated sub-ring of $F$ produced by the entries in the matrices of the generators of the group, then a coherent sequence of finite index ideals of $R$ will automatically deliver the sought coherent sequence of normal subgroups (see section~\ref{Sec:Numerics}). For example, all free groups $\FM_n$, $n \geq 2$, accept faithful representations as $2 \times 2$ matrices with integer coefficients, hence as matrices from ${\rm GL}(2,\ZM)$. Passing from the ring $\ZM$ to its ideal $p \, \ZM$, with $p$ a natural number greater than one, automatically generates the finite index normal subgroup $\FM_n \cap \widetilde{\rm GL}(2,p\ZM)$, where the latter is the subgroup of ${\rm GL}(2,\ZM)$ with elements
\begin{equation}
\begin{pmatrix} 1 + a & b \\ c & 1+ d \end{pmatrix}, \quad a,b,c,d \in p\ZM.
\end{equation}
 Likewise, $p^2 \, \ZM \subset p \, \ZM$ will produce a subgroup of finite index too, which is a subgroup of the first one. And so on. The curious reader can check Fig.~\ref{Fig:IDS1} for a numerical validation of the computer algorithm stemming from these ideas. Let us acknowledge that \cite{SchumacherAHP2015} also tested this type of finite approximations, but these tests have been presented there under the umbrella of McKay's strategy, confined to the free groups. By switching to L\"uck's framework (see section~\ref{Sec:Lueck}), however, we can be sure that same strategy works for any residually finite group and, in section~\ref{Sec:ConvFf2}, we numerically demonstrate that this is the case for Fuchsian groups. For example, our results there pass the tests of the known rigorous bounds on the spectra \cite{Bartholdi2004,GouezelCPC2015,KestenTAMS1959,KestenMS1959,
PaschkeMZ1993,ZukCM1997,BartholdiCM1997}, and they are converged enough to use them to derive simple combinatorial statements about the underling hyperbolic lattice.

Addressing now to the community of computational physicists, we briefly explain why identifying a coherent sequence of normal subgroups is equivalent to imposing periodic boundary conditions, {\it e.g.}, as it is routinely done on the familiar regular Euclidean lattices (see also \cite{MaciejkoPNAS2022} for a nice discussion of the issue), and why their convergence to the thermodynamic limit is important. First, if $N$ is a finite-index normal subgroup of an amenable group $G$,\footnote{The technical points related to non-amenability will be discussed in subsection.~\ref{Sec:GCAlg}.} then the quotient group $G/N$ comes with a canonical homomorphism that can be lifted to a morphism $\phi$ between the group $C^\ast$-algebra $C^\ast(G)$, where the exact Hamiltonian $H$ and its Green function $(z-H)^{-1}$ live, and the finite group $C^\ast$-algebra $C^\ast(G/N)$. This supplies a precious way to canonically approximate the exact Hamiltonian and its functional calculus. Indeed, notice that $\phi(f(H)) = f(\phi(H))$ for any continuous function on the real axis. Furthermore, by default, the spectrum of $\phi(H)$ is a subset of the spectrum of the exact Hamiltonian. Thus, if $H$ has a spectral gap, $\phi(H)$ will display this gap as well or, in other words, the bulk spectral gaps of $H$ are not contaminated by the finite approximation, and this is exactly what ``periodic boundary conditions'' are ought to deliver. Of course, the difficult part is to prove the convergence of the finite approximations as we climb the coherent sequence of subgroups, and this is what L\"uck delivered in  \cite{LuekGFA1994} and it is being numerically confirmed in this present work for the class of residually finite groups. In the familiar case $G=\ZM$, all these work as follows: One can choose $G_n =p^n \ZM$, $p \geq 2$, as the coherent sequence of normal subgroups. Then the Cayley graph of $\ZM/p^n \ZM = \ZM_{p^n}$ is just the closed graph shown in Fig.~\ref{Fig:DB1} and the shift operator $T$ is mapped by $\phi_n$ into the circular shift operator on this graph. Thus the mapping $H \to \phi_n(H)$ amounts to replacing the shift operator by the circular shift operator on the finite closed graph. Lastly, Eq.~\eqref{Eq:S0} assures us that $\langle k\,{\rm mod}\,p^n |(\phi_n(H)-z)^{-1}|l\,{\rm mod}\,p^n \rangle$ converges to desired limit $\langle k | (H-z)^{-1} |l\rangle$. With this, we have explained the working mechanism behind the heuristic dangling bond removal procedure.

This has to be compared with the recent work \cite{ChengPRL2022} on the same subject, where only bands in one and two dimensional representations were resolved (for hyperbolic crystals). In the same context, normal subgroups of lower index were produced in \cite{MaciejkoPNAS2022} with an existing software \cite{RoberLINS} and few projected spectra were mapped out. With a strategy based on \cite{RoberLINS}, however, it is practically impossible to generate coherent sequences of subgroups (see subsection~\ref{Sec:ConvF2}) and, perhaps for this reason, no convergence results have been presented in \cite{MaciejkoPNAS2022}. In fact, the convergence problem is not even mentioned in the physics literature on the subject. Furthermore, some of the works advertise the finite dimensional representations as genuine, but this is a sensible point because the reduced group $C^\ast$-algebras, where the Hamiltonians and their Green's functions live, of a large class of finitely generated groups, which includes the free and Fuchsian groups, are simple \cite{AkemannIUM}, hence they do not accept finite dimensional representations (see subsection~\ref{Sec:GCAlg}). By comparing with exact results, which are quite abundant in the existing literature (see section~\ref{Sec:Numerics}), we demonstrate that the finite dimensional representations produce spurious spectra (see Example~\ref{Ex:SpurSpec}) and that the genuine spectral characteristics of Hamiltonians can only be inferred from a limiting procedure on coherent sequences of finite dimensional representations. This is important not only for getting the spectral properties right, but also for getting the partition functions of the statistical physics of the crystals and for computing response functions to external stimuli. Lastly, let us mention that a universal solution when it comes to the spectral properties of aperiodic systems relies on the evaluate the local density of states at or near the center of the finite-size crystal with open boundary conditions \cite{YuPRL2020,StegmaierPRL2022}. This method, however, converges only as an inverse power with the crystal size \cite{MassattMMS2017} and its reliability when it comes to computing thermodynamic coefficients is yet to be demonstrated.

Let us mention that there is also a strong interest in aperiodic systems and their approximations. Therefore, it will be extremely interesting to investigate if and how L\"uck's formalism can be extended to groupoids and their $C^\ast$-algebras \cite{RenaultBook}.  An algebraic strategy akin to that of L\"uck \cite{LuekGFA1994,LuekBook2002} was applied to crossed product algebras by free abelian groups in \cite{ProdanARMA2013,ProdanSpringer2017}, which, together with an approximation of a non-commutative differential calculus, enabled accurate mappings of cyclic cocycle pairings with $K$-theoretic elements, in the context of disordered systems. A long term goal of ours is to extend this program to crossed products by non-commutative groups, and the results communicated in this work gives us hope that this is indeed achievable. For example, the hyperbolic lattice systems are known to posses interesting topological dynamics \cite{ComtetAP1987,CareyCMP1998,MarcolliCCM1999,MarcolliCMP2001,MathaiATMP2019} robust against disorder. Let us mention that interesting strategies involving analytic and algebraic arguments for aperiodic amenable lattices have been devised in  \cite{ElekJFA2008}. See also \cite{BeckusJFA2018,BeckusAHP2019} for additional results that mostly engage analytic arguments.

Lastly, we briefly describe the organization of our paper.  In section~\ref{Sec:CGraphs}, we first introduce the free groups and the groups that can be presented using generators and relations, with a predominant focus on how to handle them on a computer. Then we discuss the associated Cayley graphs and supply relevant examples. In section~\ref{Sec:QDynamics}, we discuss the quantum dynamics over Cayley graphs and pinpoint key characteristics of the generators of the dynamics. Then we continue with a presentation of various operator algebras associated with a group and of their relations with the actual dynamics observed over Cayley graphs. We end this section with some exact results relating the combinatorics of graphs and spectral properties of a class of operators. In section~\ref{Sec:Lueck}, we discuss the converging approximations of L\"uck \cite{LuekGFA1994,LuekBook2002} and present our point of view on the subject. We conclude with computer simulations in section~\ref{Sec:Numerics}.

\section{Finitely generated groups and their Cayley graphs}\label{Sec:CGraphs}

Finitely generated groups are groups that can be presented in terms of finite sets of generators and relations, that is, as quotients of free groups by normal subgroups (the normal closures of the relations). Thus, the point of departure for this class of groups is the class of free groups, which we consider first into some detail. A pedagogical introduction to the subject can be found in \cite{LohBook}, which we use below to state basic facts and to fix the terminology and notation.

\subsection{The free group} Up to group isomorphisms, the free group $\FM_n$ with $n$ generators $\{X_i\}_{i \in \{1,\ldots,n\}}$ is defined by the following universal property: If $f:\{1,\ldots,n\} \to G$ is any function to a group $G$, then there exists a unique group homomorphism $\varphi$ making the following diagram commute:
\begin{equation}\label{Eq:Universal1}
\begin{tikzcd}
	\ & \FM_n \\
	\{1,\ldots,n\} \ \ \  & \ \\
	\ & G
	\arrow["\alpha", from=2-1, to=1-2]
	\arrow["f", from=2-1, to=3-2]
	\arrow["\varphi", from=1-2, to=3-2]
\end{tikzcd}
\end{equation}
Here, $\alpha$ is the map $\alpha(i) = X_i$ into the generators of $\FM_n$. Among many other things, this universal property shows how easy it is to construct group homomorphisms from $\FM_n$ to any other group.\footnote{Note that there is no restriction of any kind on the map $f$.} This can and will be exploited in our numerical investigations.

More concretely, $\FM_n$ is the group with $n$ generators and no relations,
\begin{equation}
\FM_n =\FM(X_1,\ldots X_n) = \langle X_1,\ldots ,X_n\rangle.
\end{equation}
If 
\begin{equation}\label{Eq:GenSet1}
\Ss = \{X_1\ldots, X_{n-1},X_{n}\} \cup \{X_1\ldots, X_{n-1},X_{n}\}^{-1}
\end{equation}
denotes the symmetric set of generators, then, for each $k \in \NM^\times$ and map $w:\{1,\ldots,k\} \to \Ss$, there is a canonically associated element of $\FM_n$
\begin{equation}
X_w = X_{w(1)} \ldots X_{w(k)}.
\end{equation} 
One refers to both $w$ and $X_w$ as words built from the alphabet $\Ss$. They sample the entire free group $\FM_n$.  Nevertheless, to the set of these words, we add the word $w_\emptyset : = e$, the neutral element of the group.

An element from $\FM_n$ can have many different presentations in terms of the generators. Indeed, a word $X_w$ can be reduced via the substitutions 
\begin{equation}\label{Eq:Reduction}
X_{\pm i} X_{\pm j} X_{\mp j} \mapsto X_{\pm i}, \quad X_{\pm j} X_{\mp j}X_{\pm i} \mapsto X_{\pm i}, \quad X_{\pm j} X_{\mp j} \mapsto 1,
\end{equation}
or amplified via the substitutions,
\begin{equation}\label{Eq:Amplification}
X_{\pm i} \mapsto X_{\pm i} X_{\pm j} X_{\mp j}, \quad X_{\pm i} \mapsto \quad X_{\pm j} X_{\mp j}X_{\pm i}, \quad 1 \mapsto X_{\pm j} X_{\mp j},
\end{equation}
without changing the identity of the element in $\FM_n$. Here and in the following, $X_{-j}$ stands for $X_j^{-1}$. Two words are declared equivalent if they are equal in $\FM_n$ or, equivalently, if one can be generated from the other via the reductions and amplifications we just mentioned. This is a true equivalence relation \cite{LohBook}[Sec.~3.3.1] and there is a one-to-one correspondence $[w] \mapsto X_w$ between the classes of words and the elements of the free group.

\begin{definition} Given a word $w$ with letters from $\Ss$, we define the length of the word $\bm l(w)$ to be the number of its letters. In addition, we set $\bm l(w_\emptyset)=0$.
\end{definition}

Note that all elementary substitutions~\eqref{Eq:Reduction} reduce the length of a word by exactly 2 units, while the elementary substitutions~\eqref{Eq:Amplification} increase the length of a word by exactly 2 units. This results in a special property of the free groups:

\begin{proposition}[\cite{LohBook}] Each class $[w]$ has a unique word with least length.
\end{proposition}

\begin{remark}{\rm We call the word identified above the maximally reduced word of the class $[w]$ and we denoted it by $\bar w$. According to the above statement, there is a bijection between the maximally reduced words and the elements of $\FM_n$, which enables an efficient encoding of the group elements on a computer. Unfortunately, this is no longer the case when relations are present.
}$\Diamond$
\end{remark}

Given two words, $w_j : \{1,\ldots,k_j\} \to \Ss$, $j=1,2$, one considers the concatenation
\begin{equation}
w_1 \vee w_2 : \{1,\ldots, k_1+k_2\} \to \Ss, \quad (w_1 \vee w_2)(i) = \left \{
\begin{array}{ll}
 w_1(i) \ & {\rm if} \ i \leq k_1, \\
 w_2(i-k_1) \ & {\rm if} \ i > k_1.
\end{array}
\right .
\end{equation}
Then the product of the associated elements in $\FM_n$ works as
\begin{equation}
X_{[w_1]} \cdot X_{[w_2]} = X_{[w_1 \vee w_2]}.
\end{equation}
One can verify that, if $w_1 \sim w'_1$ and $w_2 \sim w'_2$ are equivalent pairs of words, then $[{w'_1 \vee w'_2}]= [{w_1 \vee w_2}]$, hence the multiplication rule written above is well defined. Also, this multiplication is associative (see \cite{LohBook}[Prop.~3.3.5]).

Given a word $w$ of length $\bm l(w)$, we can create a new word of length $\bm l(w)+1$ by concatenating letters from $\Ss$ to the front ({\it i.e.} to the left). More precisely:

\begin{proposition} If $w$ is a maximally reduced word of length $k$, then there are $2n-1$ maximally reduced words of length $k+1$ that can be produced from $w$ by concatenating one letter in front. As a consequence, there are precisely $2n(2n-1)^{k-1}$ maximally reduced words of length $k$ for $k \geq 1$ and $\frac{n}{n-1}((2n-1)^k-1)+1$ maximally reduced words of length $k$ or smaller.
\end{proposition}

\begin{figure}[t]
\center
\includegraphics[width=\textwidth]{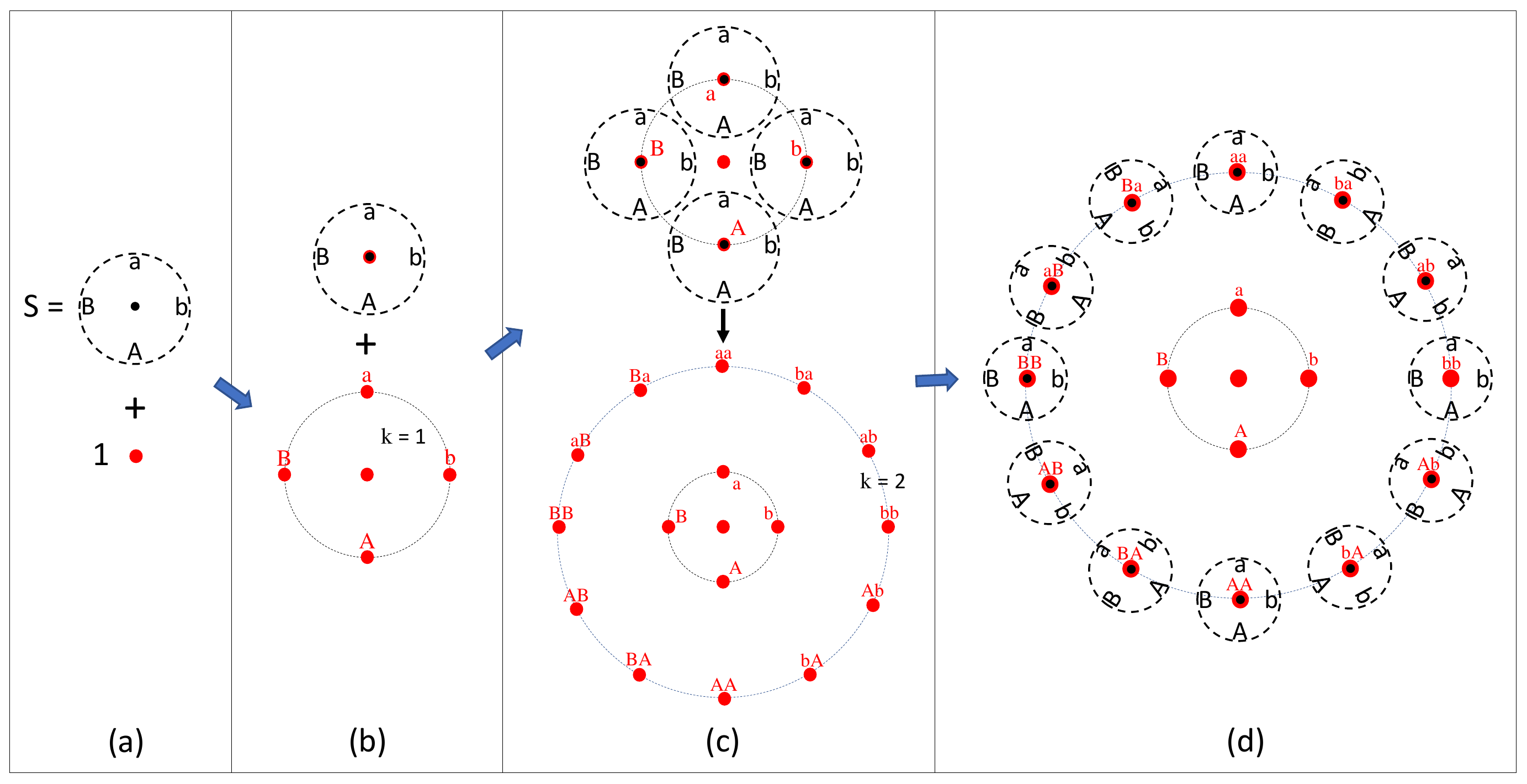}\\
  \caption{\small Algorithm for generating and storing the elements of the free group $\FM_2$. For convenience, the symmetric set $\{X_1,X_2,X_{-1},X_{-2}\}$ has been relabeled as $\{a,b,A,B\}$. The arrows indicate the flow of the algorithm.
}
 \label{Fig:GenF2}
\end{figure}

\begin{remark}\label{Re:Ratio}{\rm According to the above statement, the ratio between the number of maximally reduced words of length $k$ and that of maximally reduced words of length $k$ or smaller tends to $(2n-2)/(2n-1)$ in the limit $k\to \infty$. This confirms one of our statements saying that the ratio between the number of boundary sites and the number of total sites of a truncated lattice does not decrease to zero as the size of the truncation is increased.
}$\Diamond$
\end{remark} 

Fig.~\ref{Fig:GenF2} illustrates a practical way to systematically cycle through the set of elements of lower length and concatenate letters without producing duplicates. The key is to look at the set $\Ss$ in a circular way, {\it e.g.} as shown in Fig.~\ref{Fig:GenF2}(a) for the group $\FM_2$. There, $X_{1}$, $X_{2}$, $X_{-1}$ and $X_{-2}$ are rendered on a circle, in this sequence, while $X_0=e$ sits at the center of this circle. Then the starting point in enumerating and storing the elements of $\FM_2$ is the set $\Ss$, so presented, and the neutral element e. In the first step, shown in Fig.~\ref{Fig:GenF2}(b), one centers $\Ss$ on the neutral element and produces the elements $X_{1}$, $X_{2}$, $X_{-1}$ and $X_{-2}$, which are arranged again on a circle. At the second step, shown in Fig.~\ref{Fig:GenF2}(c), the imaginary circle hosting the set $\Ss$ is centered on a word $w_1$ of length one and then $\Ss$ is rotated such that $w_1^{-1}$ resides inside the circle hosting the words of length one. Then concatenation of the letters residing outside this circle produces maximally reduced words of length 2. By applying this protocol to the words $X_{1}$, $X_{2}$, $X_{-1}$, $X_{-2}$ of length 1, in this order, we produce all 12 maximally reduced words of length 2, in the well-defined order
$$
\begin{aligned}
& 
X_{-2}X_{1},X_1X_1,X_2X_1, 
X_1 X_2,X_2X_2,X_{-1}X_2, \\
& \qquad \qquad \qquad \qquad X_2X_{-1},X_{-1}X_{-1},X_{-2}X_{-1},
X_{-1}X_{-2},X_{-2}X_{-2},X_1X_{-2}.
\end{aligned}
$$
By applying the protocol again to the words listed above, we generate all 36 maximally reduced words of length 3, again, in a well defined order that can be read off from the diagram in Fig.~\ref{Fig:GenF2}(d). Thus, the process illustrated in Fig.~\ref{Fig:GenF2} generates all maximally reduced words of length less or equal than 3.

As always, the starting point of a simulation involving a group $C^\ast$-algebra is the encoding of the multiplication table of the group. If $\FM_n$ has been stored already as a set, {\it e.g.} using the above algorithm, then generating its multiplication table amounts to reducing the concatenations produced by pairs of words and to identifying within the set of elements the result of these reductions. Of course, the free groups are not finite, hence we can only store a finite section of the groups. In practice, one will typically generate all reduced words of length smaller or equal than a threshold, fixed by the available computational resources. Likewise, we can only store a finite section of the multiplication table of the group. In our work, we will not deal with the full group directly, but we will rather engage its finite approximation. Still, we use the full group and the algorithm described above to generate these finite approximations and to exemplify various boundary conditions in subsection~\ref{Sec:Misleading}

Another strategy for dealing with $\FM_n$ on a computer is to use a faithful representation in ${\rm GL}(n,\RM)$. For example, a useful presentation of this type can be found in \cite[Lemma 2.3.2.]{CeccheriniSpringer2010} for $\FM_2$ (see also \cite{NewmanAJM1964}):
\begin{equation}\label{Eq:F2Mat}
X_1 \mapsto {\small \begin{pmatrix} 1 & 2 \\ 0 & 1 \end{pmatrix}}, \quad X_2 \mapsto {\small \begin{pmatrix} 1 & 0 \\ 2 & 1 \end{pmatrix}}.
\end{equation}
Using the algorithm described above, one can easily produce the matrix representations of all maximally reduced words of length $k$ or smaller. Now, note that the entries of the above matrices and their inverses are integer, hence the representation actually lands in ${\rm GL}(2,\ZM)$. The ideals of the ring $\ZM$ are all of the form $p\, \ZM$, for some natural number $p$, hence it is straightforward to produce normal subgroups \cite{SchumacherAHP2015}, as already advertised in our opening remarks. The quotient groups $G_N$ can be obtained by applying ${\rm mod}\, p^N$ on the entries in the $\FM_2$ matrices, of which we only need a finite number to produce the whole $G_N$. Indeed, when the words in $\FM_2$ exceed a certain length, applying ${\rm mod}\, p^N$ will only produce duplicates.

Any other free group can be identified with a subgroup of $\FM_2$. More precisely:
\begin{proposition}[\cite{CeccheriniSpringer2010},~Corollary D.5.3.] The subgroup of $\FM_2$ generated by the subset
\begin{equation}
X = \{X_2^{i} X_1 X_{-2}^{i}, \ 0 \leq i \leq n-1\}
\end{equation}
is isomorphic with $\FM_n$.
\end{proposition}

The above statement not only assures us that all $\FM_n$ can be realized as subgroups of $GL(2,\ZM)$, but it also gives us the means to write concrete representations for the generators. For example, for $\FM_3$, we find
\begin{equation}\label{Eq:F3Mat}
X_1 \mapsto {\small \begin{pmatrix} 1 & 2 \\ 0 & 1 \end{pmatrix}}, \quad X_2 \mapsto {\small \begin{pmatrix} -3 & 2 \\ -8 & 5 \end{pmatrix}}, \quad X_3 \mapsto {\small \begin{pmatrix} -7 & 2 \\ -32 & 9 \end{pmatrix}}.
\end{equation}
Other possibilities can be found in \cite[Sec.~3.2]{MeierBook} and \cite{ZubkovMZ1998}. Therefore, coherent sequences of subgroups of $\FM_3$ or any other $\FM_n$ groups can be generated using the same principles. These aspects will be discussed in detail in section~\ref{Sec:ConvF2}.

\begin{remark}\label{Re:NonUnitGen}{\rm Note that the matrices~\eqref{Eq:F2Mat} are not unitary. This is not a concern at all, because these matrices are only used to generate and label the Cayley graphs. The translation operators over these Cayley graphs will be, by default, unitary.
}$\Diamond$
\end{remark} 

\subsection{Finitely generated groups} Given any finite set $\Gg=\{g_1,\ldots,g_n\}$, $n = |\Gg|$, we can canonically associate to it the free group $\FM(\Gg) \simeq \FM_n$. 

\begin{definition} The group generated by the finite set $\Gg$ and finite set of relations $\Rr \subset \FM(\Gg)$ is the quotient group $\FM(\Gg)/\widehat \Rr$, where $\widehat \Rr$ is the smallest normal subgroup of $\FM(\Gg)$ containing the set $\Rr$. This data is usually encoded in the notation $\langle \Gg | \Rr \rangle$.
\end{definition}

\begin{remark}{\rm We will always assume that the generating set $\Gg$ has the smallest cardinality possible. Hence, the ranks of the finitely generated groups will always coincide with the cardinal $|\Gg|$ of the generating set.
}$\Diamond$
\end{remark}

\begin{example}{\rm The hyperbolic crystals are generated from Fuchsian groups, which are discrete subgroups of the full isometry group of the hyperbolic disk or of any other equivalent model. Thus, they are the equivalent of the space groups in the context of hyperbolic geometry. Fuchsian groups with compact fundamental domain are classified~\cite{Katok1992,FordBook,MagnusBook1974} by their signature $(g,\nu_1,\ldots,\nu_r)$ and can be presented in terms of generators and relations as follows
{\small
$$
\begin{aligned}
\Ff_{g,\bm \nu} =  \big \langle a_1,b_1,\ldots,a_g,b_g,x_1,\ldots,x_r \ |  x_1^{\nu_1}, \ldots, x_r^{\nu_r},  x_1 \cdots x_r [a_1,b_1] \cdots [a_r, b_r] \big \rangle,
\end{aligned}
$$}
where $[a,b] : =ab a^{-1} b^{-1}$ denotes the commutator of two elements.
}$\Diamond$
\end{example}

\begin{remark}\label{Re:Fuch0}{\rm The Fuchsian groups $\Ff_g$ with $\bm \nu =0$ are particularly interesting because the fundamental domains of the corresponding hyperbolic crystals are surfaces $\Sigma_g$ of genus $g$ and the groups are isomorphic to the fundamental groups $\pi_1(\Sigma_g)$ of these domains. The Cayley graphs of these particular Fuchsian groups are the equivalent of Bravais lattices. Specifically, any general Fuchsian group $\Ff_{g,\bm \nu}$ has a non-unique $\Ff_{g'}$ as a normal subgroup, generating an exact sequence of groups \cite{MathaiATMP2019}
\begin{equation}
e \rightarrow \Ff_{g'} \rightarrow \Ff_{g,{\bm \nu}} \rightarrow P=\Ff_{g,{\bm \nu}}/\Ff_{g'} \rightarrow e,
\end{equation}
where $P$ is a finite group, which can be viewed as the point group of the hyperbolic crystal \cite{BoettcherPRB2022}. There is a relation between $g'$ and the index of $P$, namely, $g'=1+|P|(2(g-1)+(r-\nu))/2$ with $\nu=\sum 1/\nu_j$.
}$\Diamond$
\end{remark}

The finitely generated groups also enjoy a universal property that fixes such a group up to isomorphisms. As for the free group, this universal property supplies the means to generate group morphisms from finitely presented groups, which will be again the key to handling such groups on a computer.

\begin{proposition}[\cite{LohBook},~Prop.~2.2.18] A finitely generated group $\langle \Gg | \Rr \rangle$ and the canonical map $\pi : \Gg \to \langle \Gg | \Rr \rangle$ enjoy the following universal property: For any group $G$ and every map $f : \Gg \to G$ such that $\varphi(r) = e$ for all $r \in \Rr$, where $\varphi$ is the group morphism from Eq.~\eqref{Eq:Universal1}, there exists precisely one group homomorphism $\bar \varphi : \langle \Gg|\Rr \rangle \to G$ such that $\bar \varphi \circ \pi = f$.
\end{proposition} 

The word problems for finitely generated groups with non-trivial relations are very difficult, in general. One such problem is determining if two words from $\FM(\Gg)$ are equal in $\langle \Gg|\Rr \rangle$. Another problem is how to associate unique words to the elements of $\langle \Gg |\Rr \rangle$. For the Fuchsian groups mentioned in Remark~\ref{Re:Fuch0}, the first problem has an algorithmic solution due to Dehn \cite{DehnCollection}, while a solution for the second problem can be found in \cite{SeriesETDS1981}. Both solutions are highly relevant for the numerical simulation of the hyperbolic crystals.

In this work, however, we rely on explicit mappings of $\Ff_2$ into ${\rm GL}(2,\RM)$, which can be found in \cite{DupuyJMP} and \cite{MaskitPAMS1999}. These works deliver the matrix presentations and also the hyperbolic disk transformations. We will use the parametrization from \cite{MaskitPAMS1999}, which is simpler:
\begin{equation}\label{Eq:Ff2Mat}
\begin{aligned}
& a_1 \mapsto {\small \begin{pmatrix}
	 2 \sqrt{3}+2 & -3 \\
	 3 & 2-2 \sqrt{3} \\
	\end{pmatrix}}, & b_1 \mapsto {\small 
		\begin{pmatrix}
		 2 & -\sqrt{3} \\
		 -\sqrt{3} & 2 \\
		\end{pmatrix}}, \\
& a_2 \mapsto {\small
	\begin{pmatrix}
	 2 & -2 \sqrt{3}-3 \\
	 3-2 \sqrt{3} & 2 \\
	\end{pmatrix}},  & b_2 \mapsto {\small 
		\begin{pmatrix}
		 2 \sqrt{3}-2 & -3 \sqrt{3}-6 \\
		 6-3 \sqrt{3} & -2 \sqrt{3}-2 \\
		\end{pmatrix}}.
\end{aligned}
\end{equation}
Given the particular entries in these matrices, they actually produce a subgroup of ${\rm GL}(2,\ZM + \sqrt{3} \, \ZM)$. The ring $\ZM + \sqrt{3} \, \ZM$ has $p \ZM + \sqrt{3} \, p \ZM$ as ideals, hence there is again a straightforward way to apply the strategy outlined in our introductory remarks (see subsection~\ref{Sec:ConvFf2}).

\begin{figure}[t]
\center
\includegraphics[width=\textwidth]{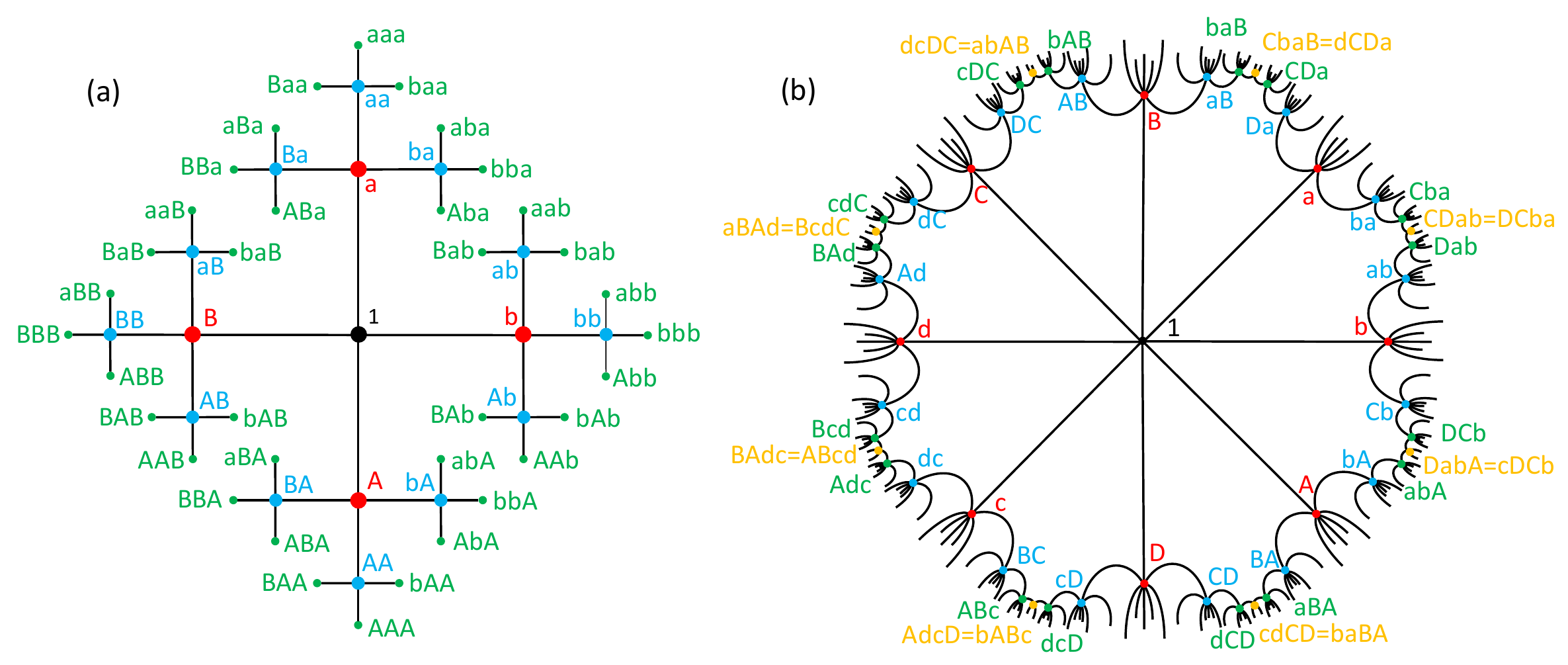}\\
  \caption{\small (a) A section of the standard Cayley graph of $\FM_2=\langle a,b\rangle$, containing all elements of length up to 3. (b) Section of the standard Cayley graph $\Cc(\Ff_2)=\langle a,b,c,d|[a,b][c,d]\rangle$ and of the equivalent surface group $\pi(\Sigma_2)$. In both graphs, the capital letters denote the inverse of the corresponding generators.
}
 \label{Fig:CayleyGraphs}
\end{figure}

\subsection{Cayley graphs}\label{Sec:CayleyGraphs}

Cayley graphs encode the data of a group in a geometric fashion \cite{MeierBook}. For example, word problems and other theoretical problems in group theory can be  solved by inspecting these geometric objects \cite[Ch.~5]{MeierBook}. On the applied side, Cayley graphs supply systematic generalizations of the crystal lattices investigated in materials science. Hence they can be an abundant source of new dynamical effects, which is our main motivation for studying them.

\begin{definition} Given a discrete group $G$ and a finite subset $S$ of $G$, the Cayley graph $\Cc(G,S)$ is the un-directed graph with vertex set $G$ and edge set containing an
edge between $g$ and $sg$ whenever $g \in G$ and $s \in S$.
\end{definition}

One should be aware that $S$ can be any finite subset of $G$ and that the geometry of the Cayley graph depends quite strongly on the choice of $S$. When the group has a standard presentation in terms of generators and relations, there is special graph $\Cc(G,\Gg)$ which we call here the standard Cayley graph and denote it simply by $\Cc(G)$. As we shall see, however, for a given model Hamiltonian, the useful Cayley graph is the one constructed from the group elements present in the expression of the Hamiltonian.

\begin{example}{\rm The standard Cayley  graph of $\FM_n$ is a regular tree with coordination $n$. Such trees are referred to in the physics literature as Bethe lattices. The standard Cayley graph of $\FM_2$ is displayed in Fig.~\ref{Fig:CayleyGraphs}(a). 
}$\Diamond$
\end{example}

\begin{example}{\rm A small section of the standard Cayley graph of the Fuchsian group $\Ff_2$ is shown in Fig.~\ref{Fig:CayleyGraphs}(b). In this case, the graph displays closed cycles, which are a reflection of the non-trivial set of relations.
}$\Diamond$
\end{example}

\begin{figure}[t]
\center
\includegraphics[width=0.7\textwidth]{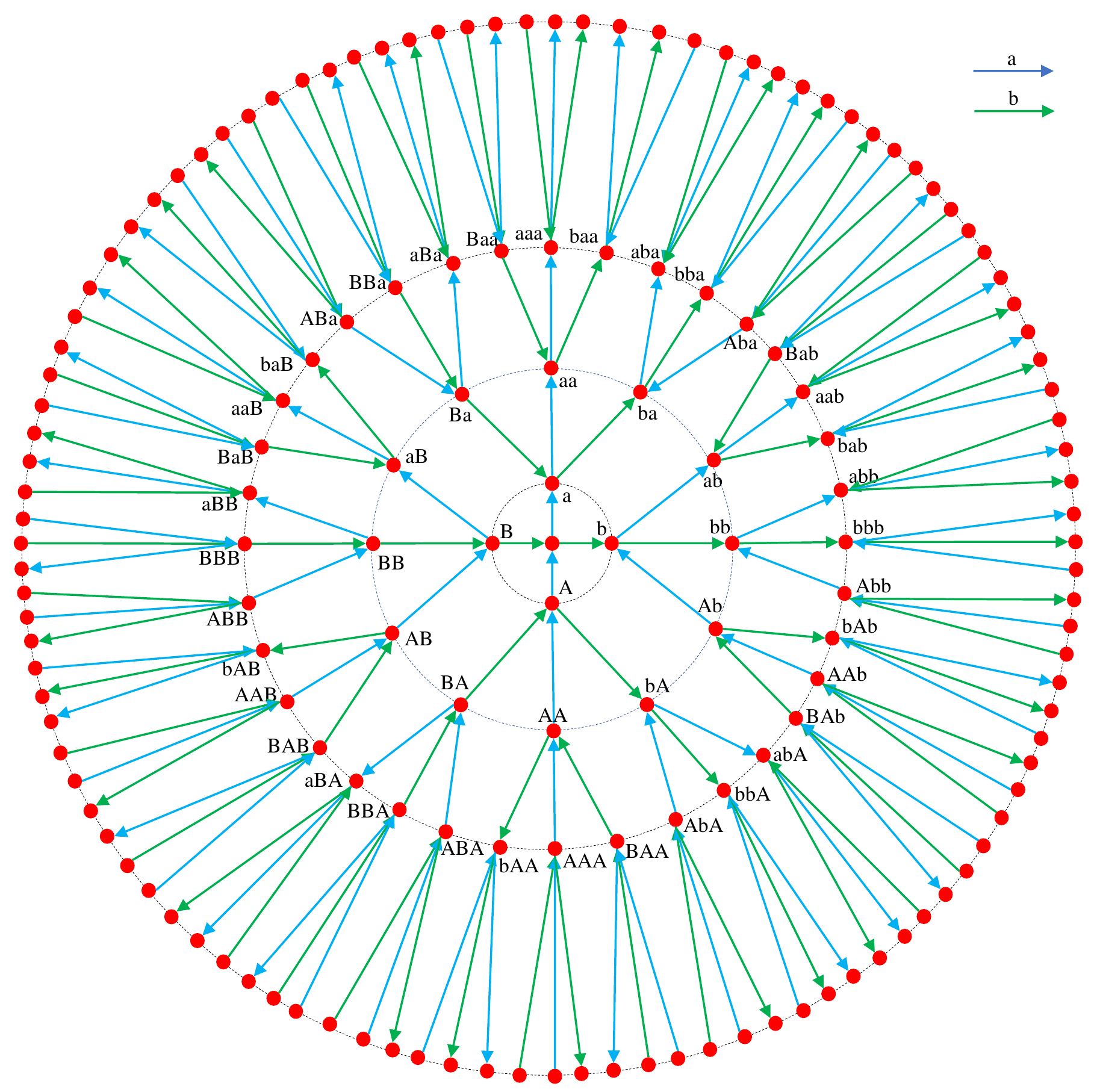}\\
  \caption{\small Section of the standard Cayley diagraph of $\FM^2(a,b)$, where $A:=a^{-1}$ and $B:=b^{-1}$. 
}
 \label{Fig:CayleyDiGraph}
\end{figure}

A more refined geometric object is the Cayley diagraph:

\begin{definition} Given a discrete group $G$ and a subset $S$ of $G$, let $c: S \to {\rm Color}$ assign a distinct color to each $s \in S$. Then the Cayley digraph $\vec \Cc(G,S,c)$ is the colored graph with vertex set $G$ and directed edges from $g$ to $sg$ for $g \in G$ and $s \in S$. All directed edges produced by $s \in S$ are assigned the color $c(s)$. 
\end{definition}

We will denote the standard diagraph of a finitely generated group by $\vec \Cc(G)$.

\begin{example}{\rm The standard Cayley diagraph of $\FM_2$ is shown in Fig.~\ref{Fig:CayleyDiGraph}.
}$\Diamond$
\end{example}

\begin{remark}\label{Re:Uniform}{\rm The standard Cayley graphs and diagraphs displayed above are distorted by our insistence to embed them in the 2-dimensional Euclidean plane. However, in the true world of these graphs, where only the connectivities count, they will appear entirely the same when observed from any of its vertices. Unfortunately, this is not always apparent from the geometric rendering of the graphs, yet it is a fact that can help us understanding the structure of these complex graphs. Note also that the standard graph of a finitely generated group is necessarily a regular graph whose degree equals the rank of the group.
}$\Diamond$
\end{remark}

The standard Cayley diagraphs reveal the flow of the vertices in response to the left or right action of the groups. For example, in Fig.~\ref{Fig:CayleyDiGraph}, one can quickly identify the orbits associated to the subgroups generated by either $a$ or $b$. Indeed, for example, the orbits of the former/latter appear as disjoint continuous paths of blue/green color. It is also clear from  Fig.~\ref{Fig:CayleyDiGraph} that there are no subsets that are invariant to the action of the full group, except for the empty set and whole graph. In other words, any finitely generated group acts ergodically on its standard Cayley graph, a fact that plays an important role in the characterization of the quantum dynamics (see Proposition~\ref{Pro:SpecChar}-ii).

\section{Quantum dynamics over Cayley graphs}\label{Sec:QDynamics}

\subsection{Patterns of quantum resonators} In our physical world, we can generate a quantum dynamics over a Cayley graph $\Cc(G)$ by simply placing quantum resonators at either the vertices or the edges of the graph. For example, in an ordinary solid state crystal,  the atoms can be viewed as quantum resonators sitting at the vertices of a graph generated from the appropriate space symmetry group. Synthetic quantum resonators, such as those in circuit quantum electrodynamics \cite{BlaisRMP2021}, can be and have been rendered and coupled in other graph configurations \cite{KollarNature2019}. 

We assume now that identical single-state quantum resonators have been placed at the vertices of a generic Cayley graph labeled by the elements of a finitely generated group $G$. In the one particle sector, the collective dynamics of the resonators takes place inside the Hilbert space $\ell^2(G)$ of square summable sequences over the graph, spanned by the vectors $|g\rangle$, $g \in G$. The dynamics is generated by a bounded and self-adjoint operator over this Hilbert space, which can always be written as a strongly convergent expansion
\begin{equation}\label{Eq:GenD}
H = \sum_{g,g' \in G} w_{g',g} \, |g' \rangle \langle g |, \quad w_{g,g'} = w_{g',g}^\ast.
\end{equation}
The parameters $w_{g',g} \in \CM$ are often called the coupling coefficients. 

\begin{remark}\label{Re:FiniteR}{\rm The physical reality always shows that the coupling coefficients become increasingly small with the graph distance between $g$ and $g'$, to a point where they cannot be resolved experimentally anymore. We recall that the graph distance between two vertices is the number of edges in the shortest path or paths joining the two vertices. Thus, the physically relevant Hamiltonians always have finite range, {\it i.e.} there exists $R \in \NM$ such that $w_{g,g'}=0$ if the graph distance between $g$ and $g'$ exceeds $R$.
}$\Diamond$
\end{remark}

If the coupling coefficients are adjusted arbitrarily, then the dynamics on the Cayley graph is no different from the dynamics over any other infinite graph. Indeed, separable infinite Hilbert spaces are all isomorphic, hence their algebras of bounded operators are isomorphic too. If we want the dynamics to reflect the symmetry of the Cayley graph, as stressed out in our Remark~\ref{Re:Uniform}, the coupling coefficients must by insensitive to the translations of the lattice:
\begin{equation}\label{Eq:Inv1}
w_{g'h,gh} = w_{g',g}, \quad \forall \ h,g,g' \in G.
\end{equation}
As we shall see, if this is indeed the case, then the dynamics of the resonators can be analyzed in a pure algebraic framework, canonically built from the group $G$ alone.

The right action of the group $G$ on itself induces the right regular representation $\pi_R$ of $G$ on $\ell^2(G)$:
\begin{equation}
U_g |g'\rangle = |g'g^{-1} \rangle, \quad g,g' \in G.
\end{equation}
One can verify that
\begin{equation}\label{Eq:Inv2}
U_hHU_h^\ast = \sum_{g,g' \in G} w_{g',g} \, |g'h^{-1} \rangle \langle g h^{-1}|= \sum_{g,g' \in G} w_{g'h,gh} \, |g' \rangle \langle g |=H,
\end{equation}
for any $h \in G$. Thus, the enforcement of the equivariant relation~\eqref{Eq:Inv1} produces Hamiltonians that are symmetric w.r.t. the natural group action. Furthermore, Eq.~\eqref{Eq:Inv1} also implies that the coefficients only depend on the relative index $q=g'g^{-1}$, which reduces the Hamiltonian to a very particular form, as the following formal manipulation shows:
\begin{equation}\label{Eq:GenD}
H = \sum_{g,g' \in G} w_{g',g} \, |g' \rangle \langle g |  = \sum_{g,g' \in G} w_{g'g^{-1},1} \, |g'g^{-1}g \rangle \langle g |,
\end{equation}
or, if we change the notation $w_{q'q^{-1},1}$ to $w_q$,
\begin{equation}\label{Eq:HFinal1}
H = \sum_{q \in G}\sum_{g \in G} w_q \, |q g \rangle \langle g | .
\end{equation}
In the above expression, we can quickly identify the left regular representation $\pi_L$ of the group on $\ell^2(G)$, induced by the left action of $G$ on itself. The conclusion is that 
\begin{equation}\label{Eq:HFinal2}
H = \sum_{q \in G} w_q \, \pi_L(q).
\end{equation}
In the light of Remark~\ref{Re:FiniteR}, the above sum involves a finite number of terms. Note however that the functional calculus with $H$ will generally produce sums with infinite terms. This is certainly the case for the resolvent $(z -H)^{-1}$.

Let us conclude with an observation about the relation Eq.~\eqref{Eq:Inv2} stating the invariance of the Hamiltonian w.r.t. the right action of the group: This relation has its origin in the associativity property of the group multiplication. Indeed,
\begin{equation}
\pi_L(q)\pi_R(h)|g\rangle = |q(gh^{-1})\rangle, \quad \pi_R(h)\pi_L(q)|g\rangle = |(qg)h^{-1}\rangle,
\end{equation}
and the associativity ensures that the two results coincide.

\subsection{Operator algebraic frameworks of analyses}
\label{Sec:GCAlg}

Given a discrete group $G$, its group algebra $\CM G$ consists of formal series
\begin{equation}\label{Eq:CG}
a = \sum_{g \in G} \alpha_g \, g, \quad \alpha_g \in \CM,
\end{equation}
where all but a finite number of terms are zero. Addition and multiplication of such formal series work in the obvious way, using the group and algebraic structures of $G$ and $\CM$, respectively. In addition, there exists a natural $\ast$-operation
\begin{equation}
a^\ast = \sum_{g \in G} \alpha_g^\ast \, g^{-1}, \quad (a^\ast)^\ast = a, \quad (\alpha a)^\ast = \alpha^\ast a^\ast, \ \alpha \in \CM.
\end{equation}
Hence, $\CM G$ is naturally a $\ast$-algebra.

\begin{remark}{\rm The group algebra has a unit $1:= e$, the neutral element of the group. Also, note that the elements of the group become unitary elements in this $\ast$-algebra, $g^\ast g = g g^\ast = 1$.
}$\Diamond$
\end{remark}

The linear map
\begin{equation}\label{Eq:Trace}
\Tt : \CM G \to \CM, \quad \Tt\Big (\sum_{g \in G} \alpha_g \, g \Big) = \alpha_e,
\end{equation}
defines a positive faithful trace on $\CM G$ and a pre-Hilbert structure on $\CM G$ via
\begin{equation}
\langle a| b \rangle : = \Tt(a^\ast b), \quad a,b \in \CM G.
\end{equation}
The completion of the linear space $\CM G$ under this pre-Hilbert structure supplies the Hilbert space $\ell^2(G)$, already encountered in the previous subsection. Indeed, one can verify that $\{|g\rangle \}_{g\in G}$ supplies an orthonormal basis for the so defined Hilbert space:
\begin{equation}
\langle g | g' \rangle = \Tt(g^{-1}g') = \delta_{g,g'}, \quad g,g' \in G.
\end{equation} 
The action of $\CM G$ on itself can be extended to the action of a bounded operator on $\ell^2(G)$, and this supplies the left regular representation $\pi_L$ of $\CM G$ inside the algebra $\BM\big(\ell^2(G)\big)$ of bounded operators over $\ell^2(G)$. Specifically,
\begin{equation}
\pi_L\Big (\sum_{g \in G} \alpha_g \, g \Big) |g'\rangle = \sum_g \alpha_g |g g' \rangle, \quad g' \in G.
\end{equation} 
This is simply the canonical extension of the left regular representation $\pi_L$ of the group to the algebra structure. By inspecting Eq.~\eqref{Eq:HFinal2}, we see that any Hamiltonian with finite coupling range can be generated as the left regular representation of an element from $\CM G$, specifically,
\begin{equation}
H = \pi_L(h), \quad h = \sum_{g \in G} w_g \, g.
\end{equation}

\begin{example}\label{Ex:Adjacency}{\rm The adjacency operator $\Delta$ of the standard Cayley graph of a finitely generated group $\langle \Gg | \Rr \rangle$ is defined as
\begin{equation}\label{Eq:Adjacency}
\Delta = \pi_L(\delta), \quad \delta = \sum_{g \in \Gg \cup \Gg^{-1}} g.
\end{equation}
It will be used extensively in our computer experiments.
}$\Diamond$
\end{example}

The group algebra $\CM G$, however, is not sufficient for analyzing the quantum dynamics. For example, the resolvent $(H -z)^{-1}$ or the projections onto the spectral bands of $H$ cannot be generated from $\CM G$ via the left regular representation. For this, we need to complete $\CM G$ to a $C^\ast$-algebra. Since the groups we are dealing with are not amenable, their group algebras accept two distinct $C^\ast$-completions \cite[Ch.~VII]{DavidsonBook}. The one that is relevant for the experiments we described in the previous subsection is the reduced group $C^\ast$-algebra $C^\ast_r(G)$ (see Remark~\ref{Re:RvsF}), generated by the completion of $\CM G$ with respect to the norm
\begin{equation}
\|a\|_r = \|\pi_L(a)\|_{\BM(\ell^2(G))},
\end{equation}
induced by the left-regular representation. The latter extends to a representation of $C^\ast_r(G)$ in $\BM(\ell^2(G))$, which coincides with the GNS representation associated to the trace $\Tt$ \cite{MurphyBook}. On the other hand, the full group $C^\ast$-algebra $C^\ast(G)$ is generated by the completion of $\CM G$ with respect to the norm
\begin{equation}
\|a\|_f = \sup \{ \|\pi(a) \|_{\BM(\Hh_\pi)} \, , \ \pi \ \mbox{is a }\ast\mbox{-representation of }\ell^1(G)\} 
\end{equation}
where $\|\cdot \|_{\BM(\Hh_\pi)}$ is the operator norm on the Hilbert space of the representation and $\ell^1(G)$ is the Banach algebra of the formal series of the type~\eqref{Eq:CG} corresponding to absolutely summable sequences$\{\alpha_g\}_{g \in G}$ (see \cite[Chs.~VII and VIII]{DavidsonBook} for more details):
\begin{equation}
\ell^1(G) = \Big \{a=\sum_{g \in G} \alpha_g \, g, \ \|a\|_1:=\sum_{g \in G}|\alpha_g| < \infty \Big \}. 
\end{equation}

As we shall see, a lot of what we are going to say in this work depends on the interplay between the algebras we just mentioned. Therefore, we elaborate on this aspect. First, by construction, $\ell^1(G)$ embeds into $C^\ast(G)$, that is, there exists an injective homomorphism $\mathfrak i : \ell^1(G) \to C^\ast(G)$. In the same time, the identity map on $\CM G$ can be extended to a $C^\ast$-algebra homomorphism 
\begin{equation}\label{Eq:Lambda}
\lambda : C^\ast(G) \to C^\ast_r(G),
\end{equation} 
which is a $C^\ast$-algebra isomporphism if and only if the group $G$ is amenable \cite[Th.~VII.2.5]{DavidsonBook}. The composition $\lambda \circ \mathfrak i$ is always injective, hence $\ell^1(G)$ also embeds into $C^\ast_r(G)$. For the free and Fuchsian groups, for example, $\lambda$ fails to be injective, which means the common elements from $\CM G$ will have different spectra when the computations are performed in $C^\ast(G)$ or $C^\ast_r(G)$. The reader can consult \cite{ArvesonBook2002} for a definition of spectrum in a pure algebraic setting and how spectra behave under algebra morphisms. Here, we mention the following phenomenon that is specific to $C^\ast$-algebras:

\begin{proposition}\label{Pro:SpecSubAlg} Let $\mathcal A$ and $\mathcal B$ be $C^\ast$-algebras such that $\Aa$ embeds into $\Bb$. Then ${\rm Spec}_\Aa(a) = {\rm Spec}_\Bb(a)$ for any $a \in \Aa$. In other words, while $a$ can be considered either as an element of $\Aa$ or an element of $\Bb$, when it comes to the spectrum, this distinction makes no difference.
\end{proposition}

We now come to some of the most powerful tools we have for group $C^\ast$-algebras, namely the functorial relation between the group and its $C^\ast$-algebra:

\begin{proposition}[\cite{NassopoulosIJCMS2008}]\label{Pro:LiftMor} Let $G$ and $H$ be discrete groups and $\rho : G \to H$ be a group homomorphism. Then $\rho$ lifts to an algebra homomorphism between $C^\ast(G)$ and $C^\ast(H)$, which on $\ell^1(G)$ acts as
\begin{equation}
\ell^1(G) \ni \sum_{g \in G} \alpha_g \, g \mapsto \sum_{g \in G} \alpha_g \, \rho(g) \in \ell^1(H).
\end{equation}
\end{proposition}

\begin{remark}{\rm According to the above statement, any morphism from $G$ to a finite group produces a finite representation of the group $C^\ast$-algebra. In fact, all finite representations of $C^\ast(G)$ can be generated in this way.
}$\Diamond$
\end{remark} 

If we replace the full by the reduced group $C^\ast$-algebra, this functorial relation continue to exist for group  morphisms with amenable kernels, in particular, for the injective ones. More precisely:

\begin{proposition}[\cite{BrownOzawaBook}, Prop.~2.5.9 \& Cor. 2.5.12]\label{Pro:CRInjective} If $\rho : H \to G$ is injective, then $\rho$ extends to an injective morphism between the corresponding reduced group $C^\ast$-algebras. Furthermore, if $C^\ast_r(H)$ is identified with its image in $C^\ast_r(G)$, then there exists a conditional expectation $E_G^H: C^\ast_r(G) \to C^\ast_r(H)$ which acts like $E_G^H(g)=0$ if $g\notin H$ and $E_G^H(g)=g$ otherwise.
\end{proposition}

\begin{corollary}\label{Cor:SpecAB} The above statement can be paired with Proposition~\ref{Pro:SpecSubAlg} to conclude that  ${\rm Spec}_{C^\ast_r(H)}(h) = {\rm Spec}_{C^\ast_r(G)}(h)$ for $h \in C^\ast_r(H)$ and $H$ a subgroup of $G$.
\end{corollary}

\begin{remark}\label{Re:RvsF}{\rm  We are now in a position from where we can explain why the physical experiments engage the reduced group $C^\ast$-algebra. Indeed, note that the trace $\Tt$ is faithful, hence the left regular representation of $C_r^\ast(G)$ is faithful. Thus $\pi_L$ is an embedding of $C^\ast_r(G)$ into $\BM(\ell^2(G))$, the $C^\ast$-algebra of bounded operators over the Hilbert space $\ell^2(G)$ of physical observations. Now, when we compute or measure the spectrum of a Hamiltonian $H = \pi_L(h)$, with $h \in \CM G \subset C^\ast_r(G)$, we do so inside $\BM(\ell^2(G))$, because  in a laboratory we manipulate and observe wavefunctions from $\ell^2(G)$. By the above proposition, however, same results will be obtained if the computations are carried inside $C^\ast_r(G)$ and this is why we claim that the physical experiments engage $C^\ast_r(G)$ and not $C^\ast(G)$.
}$\Diamond$
\end{remark}

\begin{remark}{\rm We recall that the resolvent set of an element $a$ of a $C^\ast$-algebra $\Aa$ consists of all those $\xi \in \CM$ for which $h-\xi$ is invertible in that algebra. The spectrum of $a$ is the complement of its resolvent set. To decide if $\xi$ is inside or not in this spectrum, one has to search all algebra $\Aa$ and see if there is an element that can serve as the inverse of $h-\xi$. Now, since $h$ belongs to $\CM G$ and the latter is part of both $C^\ast(G)$ and $C^\ast_r(G)$, there is the dilema in which of them to search for the inverse of $h-\xi$. The previous remark answer this question: This inverse should be searched for in $C^\ast_r(G)$. But what happens if we search for it in $C^\ast(G)$? Proposition~\ref{Pro:SpecSubAlg} requires an injective homomorphism from $\Aa$ to $\Bb$ in order to hold. As we already mentioned, the homomorphism $\lambda$ from Eq.~\eqref{Eq:Lambda} fails to be injective for the groups we study here, hence $C^\ast(G)$ does not embed into $\BM(\ell^2(G))$. This means that the search for the inverse of $h-\xi$ inside $C^\ast(G)$ will return different results and the spectrum of $H = \pi_L(h)$ will not be correctly computed.  
}$\Diamond$
\end{remark}

\begin{example}\label{Ex:SpurSpec}{\rm Consider the adjacency operator~\eqref{Eq:Adjacency} on the standard Cayley graph of $\FM_2$. It is generated by an element $\delta$ from $\CM \FM_2$, which can be seen either as an element of $C^\ast(G)$ or $C^\ast_r(G)$. Let us consider the first alternative. Since $\delta$ is a linear combination of four unitary operators, its spectrum must be contained inside the interval $[-4,4]$. Now the subgroup $[G,G]$ of commutators sits as a normal subgroup inside $G$. The quotient group $G/[G,G]$ is called the abelianization of $G$. Every 1-dimensional representation of $G$ factors through it, in the sense that it can be obtained as a composition of the projection $G \to G/[G,G]$ and a character of $G/[G,G]$. According to Proposition~\ref{Pro:LiftMor}, the morphism $G \to G/[G,G]$ lifts to a morphism between the corresponding full group $C^\ast$-algebras (but not between the reduced $C^\ast$-algebras!). Thus, we can talk about the abelianization of the adjacency operator, which is just the discrete Laplace operator on the regular lattice $\ZM^2$, whose spectrum can be easily computed as $[-4,4]$. From the behavior of the spectra under morphisms \cite{ArvesonBook2002}, we conclude that the spectrum of the adjacency operator inside $C^\ast(G)$ is the full interval $[-4,4]$. On the other hand, the spectrum of the adjacency operator inside $\BM(\ell^2(\FM_2))$, {\it i.e.} the real physical one, is the interval $[-2\sqrt{3},2\sqrt{3}]$ \cite{McKayLG1981}.
}$\Diamond$
\end{example}

\begin{remark}{\rm The above example also demonstrates that the 1-dimensional representations generate spurious spectrum. We can describe what is going on without appealing to any algebraic arguments. Indeed, as a general rule, $\xi$ belongs to spectrum of $H$ if there is a sequence $\{\psi_k\}$ from $\ell^2(G)$ such that $\|\psi_k\| = 1$ and $\lim_{k \to \infty}\|H\psi_k - \xi \, \psi_k\| =0$. Now, the abelianization of $\FM_2$ is just $\ZM^2$ and suppose that, through the 1-dimensional representations of this abelian model, we find $\xi$ and a sequence $\{\psi_k^{(1)}\}$ in $\ell^2(\ZM^2)$ which check the spectral criterion for the abelianization of $H$. Unfortunately, there is now way to pull this sequence back on $\ell^2(\FM_2)$ and check the spectral criterion for $H$ itself.
}$\Diamond$
\end{remark}

Brisk yet quite complete characterizations of the full and reduced $C^\ast$-algebras of the free group $\FM_2$ can be found in \cite[Sec.~VII.6]{DavidsonBook}. For reader's convenience, we reproduce some of these results, together with some implications relevant for the present context:
\begin{enumerate}[\noindent a)]
\item $C^\ast(\FM_2)$ accepts a discrete family $\pi_n$, $n\geq 1$, of finite dimensional representations such that $\oplus_{n\geq 1} \pi_n$ is a faithful representation. They can be produced by a standard algorithm described in \cite[pg.~204]{DavidsonBook}. On the other hand, $C^\ast_r(\FM_2)$ is simple, hence it does not accept any finite dimensional representation.
\item Both $C^\ast(\FM_2)$ and $C^\ast_r(\FM_2)$ accept faithful traces. There are many such traces on $C^\ast(\FM_2)$ but only one on  $C^\ast_r(\FM_2)$ (the one appearing in Eq.~\eqref{Eq:Trace}).
\item Both $C^\ast(\FM_2)$ and $C^\ast_r(\FM_2)$ do not not contain any non-trivial projections. As such, any Hamiltonian $h \in C^\ast_r(\FM_2)$ displays a single spectral band, hence no spectral gaps. Of course, gaps in the Hamiltonian spectra can be open by tensoring with algebras of finite matrices, which amounts to layering the lattices of metamaterials.
\end{enumerate}

\begin{remark}{\rm All a-c) points apply also to $C^\ast(\Ff_2)$ and $C^\ast_r(\Ff_2)$ algebras \cite{AkemannIUM,AkemannHJM1981}.
}$\Diamond$
\end{remark}

The following statements give a coarse characterization of the spectral properties of the elements of $C^\ast_r(G)$, for a generic finitely generated group $G$:

\begin{proposition}\label{Pro:SpecChar} Let $G$ be a finitely generated group. Then:
\begin{enumerate}[\noindent \ \rm i)]
\item If $w$ is an element of infinite order, then the spectrum of $\pi_L(w)$, which is a unitary operator by default, is the full circle and it is absolutely continuous. Same for $\pi_R(w)$.
\item If the group contains elements of infinite order, then the spectrum of any $\pi_L(h)$, $h \in C^\ast_r(G)$ is void of discrete spectrum. Consequently, the image in $\BM(\ell^2(G))$ of $C^\ast_r(G)$ through the left regular representation does not contain any compact operator.
\item If the group does not contain any element of finite order and obeys the Baum-Connes conjecture \cite{BaumContMath1994}, then the spectrum of any $\pi_L(h)$, $h=h^\ast \in C^\ast_r(G)$, is connected ({\it i.e.} there are no gaps in the spectrum).
\end{enumerate}
\end{proposition}

\begin{proof} i) We will make use of the right cosets of the subgroup $W=\langle w \rangle$ of $G$ generated by $w$. Every such right coset is invariant w.r.t. the left action of $W$ and it can be presented as $\{w^n g\}_{n\in \ZM}$, for some $g \in G$. If $C_w$ denotes the set of the right cosets of $W$, then we have the decomposition of the Hilbert space $\ell^2(G)$ into invariant sub-spaces for the action of $\pi_L(w)$, $\ell^2(G) = \bigoplus_{C_w} \ell^2(\ZM)$. Furthermore, $\pi_L(w)$ acts as the shift operator on each of the invariant subspaces $\ell^2(\ZM)$. The first statement then follows from the spectral property of this particular operator. The case involving the right regular representation is proved similarly by considering the left cosets.

ii) According to i), the right regular representation of an element of infinite order, say $w$, has continuous spectrum. Hence, there are no invariant finite dimensional subspaces of $\ell^2(G)$ for the action of $\pi_R(w)$. Therefore, $\pi_L(a)$ cannot have discrete spectrum because the projection onto its corresponding Hilbert subspace will be finite dimensional and invariant to the action of $\pi_R(w)$.

iii) In the stated conditions, $C^\ast_r(G)$ contains no proper projections \cite{BaumContMath1994}. Thus, all self-adjoint elements have gapless spectrum.
\end{proof}

\begin{remark}{\rm H. Kesten was the first to consider random walks on non-commutative groups and to investigate spectra of Markovian operators \cite{KestenTAMS1959,KestenMS1959}. He computed the spectrum of a simple random walk on a free group and the corresponding spectral measure was found to be absolutely continuous. This complements the information we already supplied in Proposition~\ref{Pro:SpecChar}, for the case of free groups. It seems that the continuity of the spectrum is a common feature of the adjacency operators on Cayley graphs of finitely generated groups \cite{MoharBLMS1989}. For infinitely generated groups, however, we already know that this is not always the case \cite{GrigorchukGD2001}. 
}$\Diamond$
\end{remark}

\begin{remark}{\rm $C^\ast$-algebras are only stable w.r.t. functional calculus involving continuous functions. When Borel calculus is mentioned, we automatically place the discussion in the context of the group von Neumann algebra, which is the bicommutant of $C^\ast_r(G)$, identified as the weak closure of the image of $C^\ast_r(G)$ in $\BM(\ell^2(G))$. We refrained from discussing this von Neumann closure for two reasons: Firstly, we are mostly interested in the continuous functional calculus ({\it e.g.} the gap projections of a Hamiltonian can be produced by this calculus) and, secondly, because $K$-theories generally become irrelevant once one passes from $C^\ast$-algebras to their von Neumann closures.
}$\Diamond$
\end{remark}

\subsection{Exact results} Here we show that there exists a large class of Hamiltonians whose spectral properties can be deduced from combinatorial exercises on appropriate Cayley graphs. Some of the results reported in this subsection are already known for the adjacency operator (see \cite{MoharBLMS1989,WoessBLMS1994} for surveys of the topic).

Now, let $S$ be a finite symmetric subset of $G$ ({\it i.e.} $S^{-1}=S$), not containing the neutral element, and consider the self-adjoint element
\begin{equation}\label{Eq:HamS}
h=\sum_{s \in S} s \in \CM G.
\end{equation}
The left regular representation of elements of this type supply the class we are interested in. The Hamiltonian~\eqref{Eq:HamS} and its associated functional calculus can be naturally studied on the Cayley graph $\Cc(G,S)$. For this particular graph, we denote by $N_{g}(n)$ the number of graph paths of length $n$ from the site labeled by $g$ to the origin labeled by the neutral element $e$.

\begin{proposition} $N_g(n) \leq |S|^n$.
\end{proposition}

\begin{proof} We observe that $h^n$, after expansion, produces the sum of all un-reduced  words of length $n$ that can be formed with letters from $S$. We also recall that $\Tt$ returns the coefficient of the neutral element, hence $\Tt(w g)$, with $w$ a word as we just described, returns one if $wg=e$ and zero otherwise. The conclusion is that $\Tt(h^n g)$ counts all the un-reduced words of length $n$ that can be generated from $S$ such that $w g =e$. These words are in one to one correspondence with the graph paths of length $n$ starting at $g$ and ending at $e$. Then
\begin{equation}\label{Eq:Ng0}
N_g(n) = \Tt(h^n \, g) = \langle e | \pi_L(h)^n |g\rangle
\end{equation}
and the seen matrix element can be bounded by $\|h\|_r^n$. The statement follows because each $s \in S$ is represented by a unitary element, hence $\|h\|_r \leq |S|$.
\end{proof}

\begin{proposition} The sequence $\{N_g(n)\}_{n\in \NM}$ accepts a generating function, that is, a function $F_g(x) : \RM \to \RM$ given by a finite linear combination of right-continuous monotone non-decreasing functions, each of them constant outside the interval $J_S=\big[-|S|,|S|\big]$, and
\begin{equation}
N_g(n) = \int_\RM x^n \, {\rm d} F_g(x), \quad \forall n \in \NM.
\end{equation}
\end{proposition} 

\begin{proof} Since $\|h\|_r \leq |S|$, $(z-h)^{-1}$ accepts the norm convergent series in $C^\ast_r(G)$ for $|z| > |S|$,
\begin{equation}
(z - h)^{-1} = \frac{1}{z} \sum_{n=0}^\infty (h/z)^n ,
\end{equation}
and same applies to its matrix elements. Therefore,
\begin{equation}
\langle e |\pi_L(z - h)^{-1} |g \rangle = \frac{1}{z} \sum_{n=0}^\infty z^{-n}\, \Tt\big (h^n \, g\big )= \frac{1}{z} \sum_{n=0}^\infty z^{-n} N_g(n).
\end{equation}
From spectral theory of self-adjoint operators, we know that the left hand side is an analytic function on the domain $\CM \setminus {\rm Spec}_{C^\ast_r(G)}(h)$, hence the right side can be analytically continued over the same domain.  Furthermore, using the residue theorem,
\begin{equation}
N_g(n) = \frac{1}{2 \pi \imath}\int_\gamma {\rm d}z \; z^n \, \langle e |\pi_L(z - h)^{-1} |g \rangle,
\end{equation}
where $\gamma$ is a curve in the complex plane encircling the domain $|z| < |S|$. But this curve can be deformed to surround the interval $J_S$ infinitely close, and the result is
\begin{equation}
N_g(n) = \int_{J_S} {\rm d}x \; x^n \, \tfrac{1}{\pi}{\rm Im}\langle e |\pi_L(x+\imath 0^+ - h)^{-1} |g \rangle.
\end{equation}
The right hand side can be expressed in terms of the family of spectral projections $\{\bm E(x)|x \in \RM\}$ of $\pi_L(h)$. Indeed, if we recall Stone's formula,
\begin{equation} 
\langle e |\bm E (x) |g\rangle = \int_{-\infty}^{x+0^+} {\rm d} t \; \tfrac{1}{\pi}{\rm Im}\langle e |\pi_L(t+\imath 0^+ - h)^{-1} |g \rangle, 
\end{equation}
then we have
\begin{equation}
N_g(n) = \int_{J_S}  x^n \, {\rm d} \langle e |\bm E(x) |g \rangle.
\end{equation}
Lastly, we can use the polarization identity to re-write
\begin{equation}\label{Eq:EV}
\langle e |\bm E(x) |g \rangle =\tfrac{1}{4}\sum_{v \in X_g} \langle v |\bm E(x) |v \rangle, \quad X_g=\{e \pm g, e \pm \imath g\} \subset \ell^2(G),
\end{equation}
hence as a finite linear combination of monotone non-decreasing functions that are right-continuous.
\end{proof}

\begin{remark}{\rm For a free group, $N_g(n)$ and $F_g(x)$ were computed explicitly in \cite{McKayLG1981} using combinatorial techniques, for $g=e$ and $S$ the generating set defined in Eq.~\eqref{Eq:GenSet1}, in which case $h$ coincides with the adjacency operator from Eq.~\eqref{Eq:Adjacency} for the standard Cayley graph of $\FM_n$. They take the form $N_e(n)=0$ for $n ={\rm odd}$,
\begin{equation}\label{Eq:Ng1}
N_e(n) =v \sum_{k=0}^{n/2} \begin{pmatrix} n \\ k \end{pmatrix} \frac{n-2k}{n}(v-1)^k, \quad n = {\rm even},
\end{equation}
and
\begin{equation}\label{Eq:ExactF}
F_e(x) = \frac{1}{2} + \frac{v}{2\pi} \left [ \arcsin \frac{x}{2 \sqrt{v-1}}-\frac{v-2}{v} \arctan \frac{(v-2)x}{v\sqrt{4(v-1)-x^2}}\right ]
\end{equation}
inside the interval $[-2\sqrt{v-1},2\sqrt{v-1}]$, $F(x)=0$ if $x <-2\sqrt{v-1}$ and $F(x)=1$ if $x>2\sqrt{v-1}$. Here, $v$ is the degree of the corresponding Cayley graphs. In particular, for the free group with two generators, $v=4$. Furthermore, $N_g(n)$ can be easily derived from $N_e(n)$ when $g$ is drawn from the set of generators:
\begin{equation}\label{Eq:Ng2}
N_{X_1^{\pm 1}}(n) = N_{X_2^{\pm 1}}(n)=N_e(n+1)/v, \quad n = {\rm odd},
\end{equation}
and $N_{X_1^{\pm 1}}(n) = N_{X_2^{\pm 1}}(n)=0$ for $n={\rm even}$. These exact results will be used to validate our numerical algorithms in subsection~\ref{Sec:ConvF2} .
}$\Diamond$
\end{remark}

\begin{corollary} Let $H = \pi_L(h)$, with $h$ as in Eq.~\eqref{Eq:HamS}, and  $z \notin {\rm Spec}(H)$. Then the matrix elements of the Green's function $(z - H)^{-1}$ can be computed as: 
\begin{equation}
\langle g | (z - H)^{-1}|g'\rangle = \int \frac{{\rm d}F_{g'g^{-1}}(x)}{z-x}.
\end{equation}
Furthermore, if $\varphi$ is a Borel function over the real line, then
\begin{equation}
\langle g|\varphi(H) |g'\rangle= \int \varphi(x) {\rm d} F_{g'g^{-1}}(x).
\end{equation}
In particular, the spectral density function can be computed as
\begin{equation}
F(x) : = \langle e|\bm E(x)|e \rangle = F_e(x).
\end{equation}
\end{corollary}

The above statements, together with the known expression of $F_e$ from Eq.~\eqref{Eq:ExactF}, supply the means to validate our numerical simulations of $\FM_2$-crystals, reported in subsection~\ref{Sec:ConvF2}. Note, however, that the results of this section can be also used in the opposite direction. Indeed, if one is in possession of rapidly convergent simulations, then the combinatorial  numbers $N_g(n)$ can be computed either via Eq.~\eqref{Eq:Ng0} or as
\begin{equation}\label{Eq:Ng01}
N_g(n)=\lim_{N \to \infty} \int x^n \, {\rm d} \langle e|\bm E_N(x) | \phi_N(g) \rangle,
\end{equation}
where $\bm E_N(x)$ are the spectral families of projections for the finitely approximated Hamiltonian and the limit is over the size of the approximation (see section~\ref{Sec:Numerics}). We will see in section~\ref{Sec:Numerics} that this is possible to some extent with the computer algorithms supplied in this work. When $N_g$ is known, Eq.~\eqref{Eq:Ng01} supplies a test on the quality of the approximated spectral measure.

\section{Converging finite approximations}\label{Sec:Lueck}

\subsection{Coherent sequences of normal subgroups}\label{Sec:CNS}

\begin{definition} A group $G$ is called residually finite if, for all $g \in G$ such that $g \neq 1$, there
exists a normal subgroup $N \triangleleft G$ such that $g  \neq N$ and $G/N$ is finite. Equivalently,
we could characterize this property by saying that for all $g \in G$ such that $g \neq 1$,
there exists a finite group $H$ and a homomorphism $\theta : G \to H$ such that $\theta(g) \neq 1$ in H.
\end{definition}

\begin{example}{\rm As stated in our introductory remarks, any finitely generated group possessing a faithful representation into $GL(n, F)$ for $F$ a field is residually finite \cite[Th.~4.2]{WehrfritzBook}. In particular, the free and Fuchsian groups are residually finite.
}$\Diamond$
\end{example}

Any residually finite group $G$ accepts a  coherent tower of finite index normal subgroups
\begin{equation}\label{Eq:NSeq}
\begin{tikzcd}
	G=G_0 & G_1 &  G_{2} & \cdots
		\arrow[rightarrowtail,from=1-2, to=1-1,"\ \mathfrak i_{1}"']
		\arrow[rightarrowtail ,from=1-3, to=1-2,"\ \mathfrak i_{2}"']
	\arrow[rightarrowtail ,from=1-4, to=1-3,"\ \mathfrak i_{3}"']
\end{tikzcd}, \quad \bigcap G_N =\{e\}.
\end{equation}
A coherent tower of objects (in a category) is a diagram as above, with the arrows representing injective morphisms. In turn, this supplies an inverse coherent sequence (i.e. with arrows representing surjective morphisms) of finite groups, 
\begin{equation}\label{Eq:NSeq}
H_N : =  G/G_N, \quad \begin{tikzcd}
	H_1 & H_2 & H_{3} & \cdots
		\arrow[twoheadrightarrow,from=1-2, to=1-1,"\ \mathfrak j_{1}"']
		\arrow[twoheadrightarrow ,from=1-3, to=1-2,"\ \mathfrak j_{2}"']
	\arrow[twoheadrightarrow ,from=1-4, to=1-3,"\ \mathfrak j_{3}"']
\end{tikzcd}
\end{equation}
with the epimorphisms given by
\begin{equation}\label{Eq:ProjT}
\mathfrak j_N : H_{N+1} \twoheadrightarrow H_{N}, \quad \mathfrak j_N(g\cdot G_{N+1}) = g \cdot G_{N}, \quad \forall \ g \in G. 
\end{equation}
They are well defined because, if $g\cdot G_{N+1} = g'\cdot G_{N+1}$, then there exists $f \in G_{N+1}$ such that $g'=gf$. Since $G_{N+1}$ embeds in $G_{N}$, $f$ belongs to $G_N$ and $g\cdot G_{N} = g'\cdot G_{N}$. Additionally, there exist the quotient morphisms $\phi_N : G \to H_N$, $\phi_N(g) = g \cdot G_N$, and the whole emergent algebraic structure can be summarized by the commutative diagram
\begin{equation}
\begin{tikzcd}
	\cdots & H_{N-1} & H_N & H_{N+1} & \cdots \\
	 &  & G &  &
	\arrow[twoheadrightarrow,from=1-2, to=1-1]
	\arrow[twoheadrightarrow,from=1-3, to=1-2,"\ \mathfrak j_{N-1}"']
	\arrow[twoheadrightarrow,from=1-4, to=1-3,"\mathfrak j_{N}"']
	\arrow[twoheadrightarrow,from=1-5, to=1-4]
	\arrow[from=2-3, to=1-2,"\phi_{N-1}" description]
	\arrow[from=2-3, to=1-3,"\phi_{N}" description]
	\arrow[from=2-3, to=1-4,"\phi_{N+1}" description]
\end{tikzcd}
\end{equation}

The category of groups is complete, hence the projective tower~\eqref{Eq:ProjT} has an inverse limit $\bar G = \varprojlim H_N$, which in general is much larger than the original group $G$. More precisely, $\bar G$ is a profinite group and, from the universal property of this limit, we can be sure that there exists a group morphism $\phi : G \to \bar G$, mapping $G$ onto a dense sub-space of $\bar G$ \cite[Lemma~1.1.7]{RibesZalesskiiBook2000}. Furthermore, the condition $\bigcap G_N = \{e\}$ implies that $\phi$ is injective. This important observation will be paired with the statements from Proposition~\ref{Pro:CRInjective} and Corollary~\ref{Cor:SpecAB}, as well as another result by L\"uck, in order to give an algebraic interpretation of the finite approximations discussed next.

\subsection{Approximation results of L\"uck} The group morphisms $\phi_N$ induce group algebra morphisms $\phi_N : \CM G \to \CM H_N$, for which we will use the same symbols. Below, we reproduce the original result by L\"uck \cite{LuekGFA1994}, in a slightly updated formulation supplied in the monograph \cite[Ch.~13]{LuekBook2002}. Throughout this subsection, we work in the setting described in the previous subsection.

\begin{proposition}[\cite{LuekGFA1994,LuekBook2002}]\label{Pro:Lueck1} Let $f \in \CM G$ be a self-adjoint and positive element and let $\{\bm E(x)|x \in [0,\infty)\}$ be the right continuous family of spectral projections of $\pi_L(f)$ in $\BM(\ell^2(G))$. Consider the associated spectral density function
\begin{equation}
F:[0,\infty) \to [0,1], \quad x \mapsto \langle e |\bm E(x) |e \rangle 
\end{equation}
and define $F_N$ in the same way for $\phi_N(f) \in \CM G_N$. Then
\begin{equation}
F(x) = \liminf F_N(x+\imath 0^+ ) = \limsup F_N(x+\imath 0^+).
\end{equation}
\end{proposition}

The above statement assures us that the diagonal matrix elements 
\begin{equation}
\langle g|\varphi(\pi_L(f))|g\rangle = \int \varphi(x) {\rm d} F(x)
\end{equation} 
can be computed with arbitrary precision using the finite approximations, for any Borel function $\varphi$. In fact, one can says so much more:

\begin{corollary}\label{Cor:FuncCalc} Consider the settings of Proposition~\ref{Pro:Lueck1}. Then the off-diagonal matrix elements $\langle e | \pi_L(\varphi(f)) |g\rangle$ can be also computed with arbitrary precision from the finite approximations, for any complex valued continuous function $\varphi$.
\end{corollary}
\begin{proof} Since any complex valued function can be written as a linear combination of real positively valued continuous functions, it is enough to assume that $\varphi$ is such a function. Furthermore, it is known that the matrix element we want to calculate can be approximated with arbitrary precision if we replace $\varphi$ by polynomials. Now, let $p(x)$ be a polynomial approximating $\sqrt{\varphi}$.  We have
\begin{equation}
\langle e |\pi_L(p^2(f))|g\rangle =  \langle \pi_L(p(f)) e |\pi_L(p(f)) g\rangle
\end{equation}
and, by using the polarization identity, the above matrix element can be written as a sum involving the terms
\begin{equation}
\langle \pi_L(p(f)) e \pm \pi_L(p(f))g|\pi_L(p(f)) e \pm \pi_L(p(f))g\rangle = \langle e\pm g |\pi_L(p^2(f)) | e \pm g\rangle
\end{equation}
and
\begin{equation}
\langle \pi_L(p(f)) e \pm \imath \pi_L(p(f)) g|\pi_L(p(f)) e \pm \imath \pi_L(p(f))g\rangle = \langle e\pm \imath g |\pi_L(p^2(f)) | e \pm \imath g\rangle.
\end{equation}
The conclusion is that $\langle e |\pi_L(p^2(f))|g\rangle$ can be expressed in terms of
\begin{equation}
\langle e | \pi_L(v^\ast p^2(f) v)|e\rangle, \quad v \in \{e \pm g, e \pm \imath g\},
\end{equation}
and each $v^\ast p^2(f) v$ is a self-adjoint and positive element of $\CM G$. As such, the diagonal elements of their left regular representations can be computed with arbitrary precision using their finite approximations 
\begin{equation}
\phi_N(v^\ast p^2(f) v) = \phi_N(v^\ast) \, p^2(\phi_N(f)) \, \phi_N(v).
\end{equation} 
Using the polarization identity in reverse, we conclude that $\langle e |p^2(\phi_N(f)) | \phi_N(g)\rangle$ can approximate the original matrix element we started with arbitrary precision.
\end{proof}

Given the above statements, it is now justified to say that the entire continuous functional calculus in $C^\ast_r(G)$ with self-adjoint elements from $\CM G$ can be reproduced with arbitrary precision using the finite approximations supplied by the coherent sequence of subgroups $\{G_N\}$. In particular, the resolvent $(z -h)^{-1}$ can be numerically computed with arbitrary precision for any $z \notin {\rm Spec}(h)$.

\subsection{An operator algebraic viewpoint} Looking back at the statement of Proposition~\ref{Pro:Lueck1}, it is impossible not to notice its pure algebraic nature, even though its proof was achieved by hard analysis. From the applied point of view, this is an extremely useful feature because, once we discover a coherent sequence of normal subgroups, there is no need for any other checks, like specific bounds on operators, etc.. Here, we describe a possible mechanism behind this feature, which can be of some guidance when dealing with finite approximations. For this, we reproduce a very general result concerning inverse limits:  

\begin{proposition}[\cite{LuekBook2002},~Th.~13.31]\label{Pro:Lueck2} Let $\{H_N\}$ be any projective tower of groups and let $\bar G$ be its limit and $\bar \phi_N : \bar G \to H_N$ be its structure functions. Let $\bar f$ be a positive element from $\CM \bar G$ and $\{\bar{\bm E}(x)|x \in [0,\infty)\}$ be the right-continuous family of spectral projections of $\pi_L(\bar f)$ in $\BM(\ell^2(\bar G))$. Consider the associated spectral density function
\begin{equation}
\bar F:[0,\infty) \to [0,1], \quad x \mapsto \langle e |\bar{\bm E}(x) |e \rangle
\end{equation}
and define $\bar F_N$ in the same way for $\bar \phi_N(\bar f) \in \CM H_N$. Then
\begin{equation}
\bar F(x) = \liminf \bar F_N(x+\imath 0^+ ) = \limsup \bar F_N(x+\imath 0^+).
\end{equation}
\end{proposition}

Definitely, Corollary~\ref{Cor:FuncCalc} can be straightforwardly adapted to the settings of Proposition~\ref{Pro:Lueck2}. Thus, the continuous functional calculus in $C^\ast_r(\bar G)$ can be reproduced with arbitrary precision from the continuous functional calculus in $C^\ast_r(H_N)$. Now, let $\{H_N\}$ be the projective tower of normal subgroups of the residually finite group $G$, introduced in subsection~\ref{Sec:CNS}, and $f$ be a positive element of $\CM G$. Then the coherent sequence $\{ \phi_N(f) \}$ defines an element of $\CM \bar G$. We recall that $G$ embeds in $\bar G$, which implies that $C^\ast_r(G)$ embeds in $C^\ast_r(\bar G)$ via the canonical map $\phi$ (cf. Proposition~\ref{Pro:CRInjective}). Obviously, the coherent sequence $\{ \phi_N(f) \}$ coincides with $\phi(f)$, hence $\{ \phi_N(f) \}$ lands in the image of $C^\ast_r(G)$ inside $C^\ast_r(\bar G)$. Therefore, the continuous functional calculus of $f$ in $C^\ast_r(G)$ coincides with the continuous functional calculus of $\{ \phi_N(f) \}$ in $C^\ast_r(\bar G)$. But the latter can be reproduced with arbitrary precision from the functional calculus with the finite approximations $\bar \phi_N(\{\phi_N(f)\})$, which coincide with $\phi_N(f)$. This is precisely the statement of Corollary~\ref{Cor:FuncCalc}.

In the light of the above arguments, we see the engines behind the approximation results as being: 1) The general fact that they hold for any inverse system of groups, and 2) The functorial properties of the reduced group $C^\ast$-algebras.

\section{Numerical simulations}\label{Sec:Numerics}

In this section we describe the principles of our computer algorithms and present results without entering too much into the actual details of the codes. The latter can be downloaded at \cite{FabianGH}, together with details on how to operate it.

\begin{figure}[t]
\center
\includegraphics[width=\textwidth]{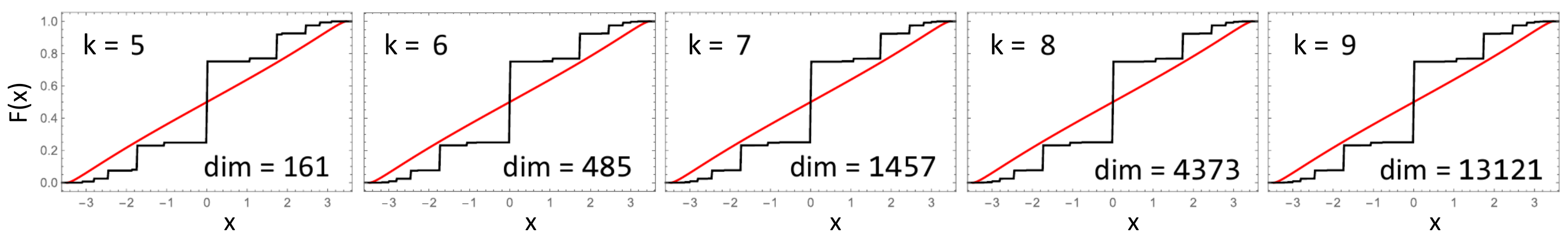}\\
  \caption{\small Eigenvalue counting (black line) for the finite size truncations with open boundary conditions, compared with the exact spectral density (red line) of the adjacency matrix $\Delta$ (see Eqs.~\eqref{Eq:Adjacency} and \eqref{Eq:ExactF}). The number $k$ indicates the word-length at which the truncation was made.
}
 \label{Fig:IDS0}
\end{figure}

\subsection{Misleading boundary conditions}\label{Sec:Misleading} Before we discuss the converging finite approximations, we exemplify what happens when non-specialized boundary conditions are used. In Fig.~\ref{Fig:IDS0}, we report the spectral densities of the finite approximations of the adjacency operator $\Delta$ introduced in Example~\ref{Ex:Adjacency}, generated with Dirichlet (or open) boundary conditions after generation $k=5,\ldots,9$ of the Cayley graph is completed. We recall that, for finite models, the spectral density of a Hamiltonian $H$ evaluated at $x \in \RM \ {\rm Spec}(H)$ amounts to counting the eigenvalues of $H$ below $x$ and dividing the result by the dimension of the Hilbert space. As seen in Fig.~\ref{Fig:IDS0}, the spectral density converges as $k$ is increased, but to a wrong limit. The explanation is that the open boundary of the graph introduces spurious spectrum, in particular, a large number of zero modes. As the size of the truncated graph grows, the spectral density becomes entirely dominated by the spurious spectrum introduced by the open boundary and this is why we see a convergence in Fig.~\ref{Fig:IDS0}. 

We can also try some ad-hoc ``periodic boundary conditions'' by folding the graph into itself. The simplest way to fold the graph, is to close the orbits discussed in subsection~\ref{Sec:CGraphs} and seen in Fig.~\ref{Fig:CayleyDiGraph}. For example, the blue arrow coming out of the node $aaa$ will end in the node $AAA$, the green arow coming out from the node $bba$ will end in the node $BBa$, while the blue arrow coming out of the node $baa$ end in the same node $baa$, and so on. Note that these orbit foldings can be performed for any truncated lattice and, as such, we can investigate what happens as the size of the truncated lattice is increased. At the first sight, this seems like a correct choice because now the generators $X_1$ and $X_2$ act as permutations of the truncated Cayley graph and, due to the universality property of $\FM_2$, the whole $\FM_2$ group is being mapped into a subgroup $P$ of the permutation group of the truncated lattice. From the functorial properties, this lifts to a morphism $\rho$ between the full group $C^\ast$-algebras and, for example, we can be sure that ${\rm Spec}_{C^\ast(P)}(\rho(\delta)) \subset {\rm Spec}_{C^\ast(\FM_2)}(\delta)$, where $\delta$ was introduced in Example~\ref{Ex:Adjacency}. The problem is, however, that 
\begin{equation}
{\rm Spec}_{C^\ast(\FM_2)}(\delta) \neq {\rm Spec}_{C^\ast_r(\FM_2)}(\delta) = {\rm Spec}(\Delta).
\end{equation} 
Indeed, the former consists of the whole interval $[-4,4]$, while the latter consists of the interval $[-2\sqrt{3},2\sqrt{3}]$ (see Example~\ref{Ex:SpurSpec}). As a result, the ad-hoc periodic boundary condition is of little use and, in fact, it can lead to misleading results. This is confirmed by the simulations reported in Fig.~\ref{Fig:IDS00}, where we can see clearly that the spectrum computed with the ad-hoc periodic boundary condition spills over the expected spectral interval. This could have been even anticipated beforehand, because the ad-hoc periodic boundary conditions produce a growing number of closed loops with the increase in the size of the truncated graph and, according to \cite{McKayLG1981}, this is exactly the scenario where the eigenvalue counting does not converges to the correct result. One can also think of the spectra seen in Fig.~\ref{Fig:IDS00} as corresponding to sequences of (not faithful) finite representations of $\delta$, with the latter seen as an element of the full group $C^\ast$-algebra. 

\begin{figure}[t]
\center
\includegraphics[width=\textwidth]{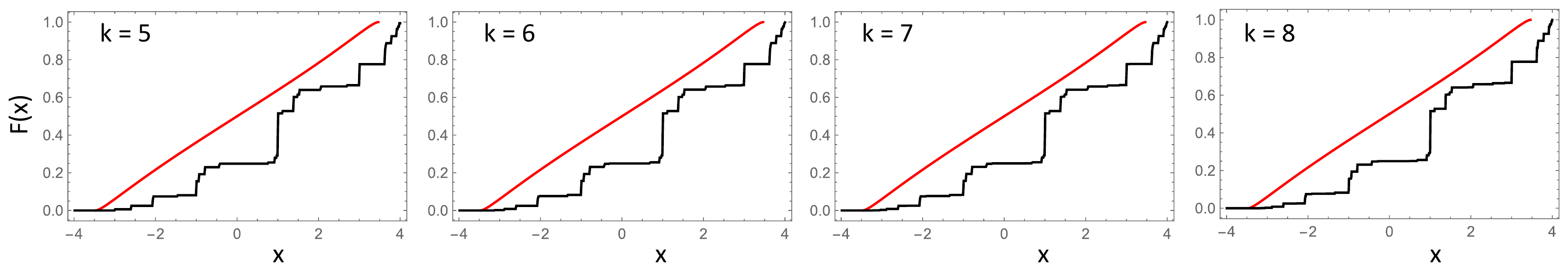}\\
  \caption{\small Same as Fig.~\ref{Fig:IDS0} but for the ad-hoc periodic conditions described in the text. 
}
 \label{Fig:IDS00}
\end{figure}

\begin{figure}[b]
\center
\includegraphics[width=\textwidth]{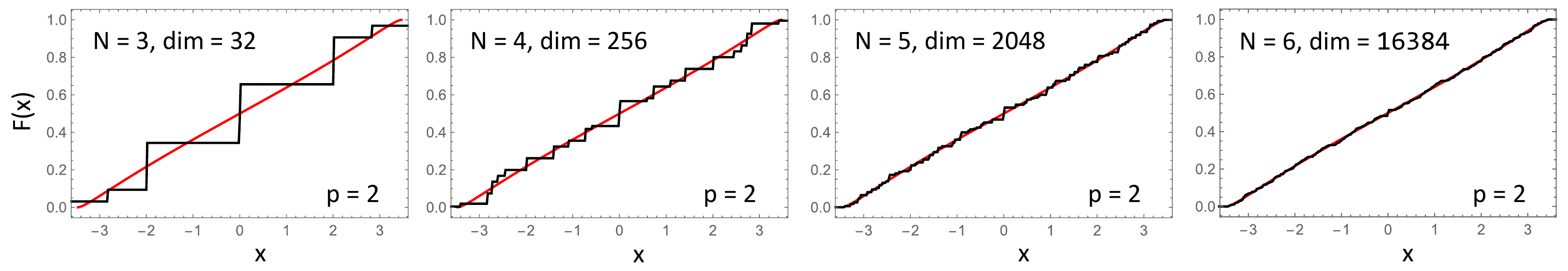}\\
  \caption{\small Spectral densities of the finite approximations $\pi_L(\phi_N(\delta))$ obtained with the converging periodic boundary conditions \eqref{Eq:ConvApp1} with $p=2$, compared with the exact spectral density (red line) of the adjacency matrix. Next to the values of $N$, we show the dimensions of the quotient groups, which determine the dimensions of the Hilbert spaces of the finite approximations.
}
 \label{Fig:IDS1}
\end{figure}

\subsection{Converging finite approximations for a free group}\label{Sec:ConvF2}  

\begin{figure}[t]
\center
\includegraphics[width=\textwidth]{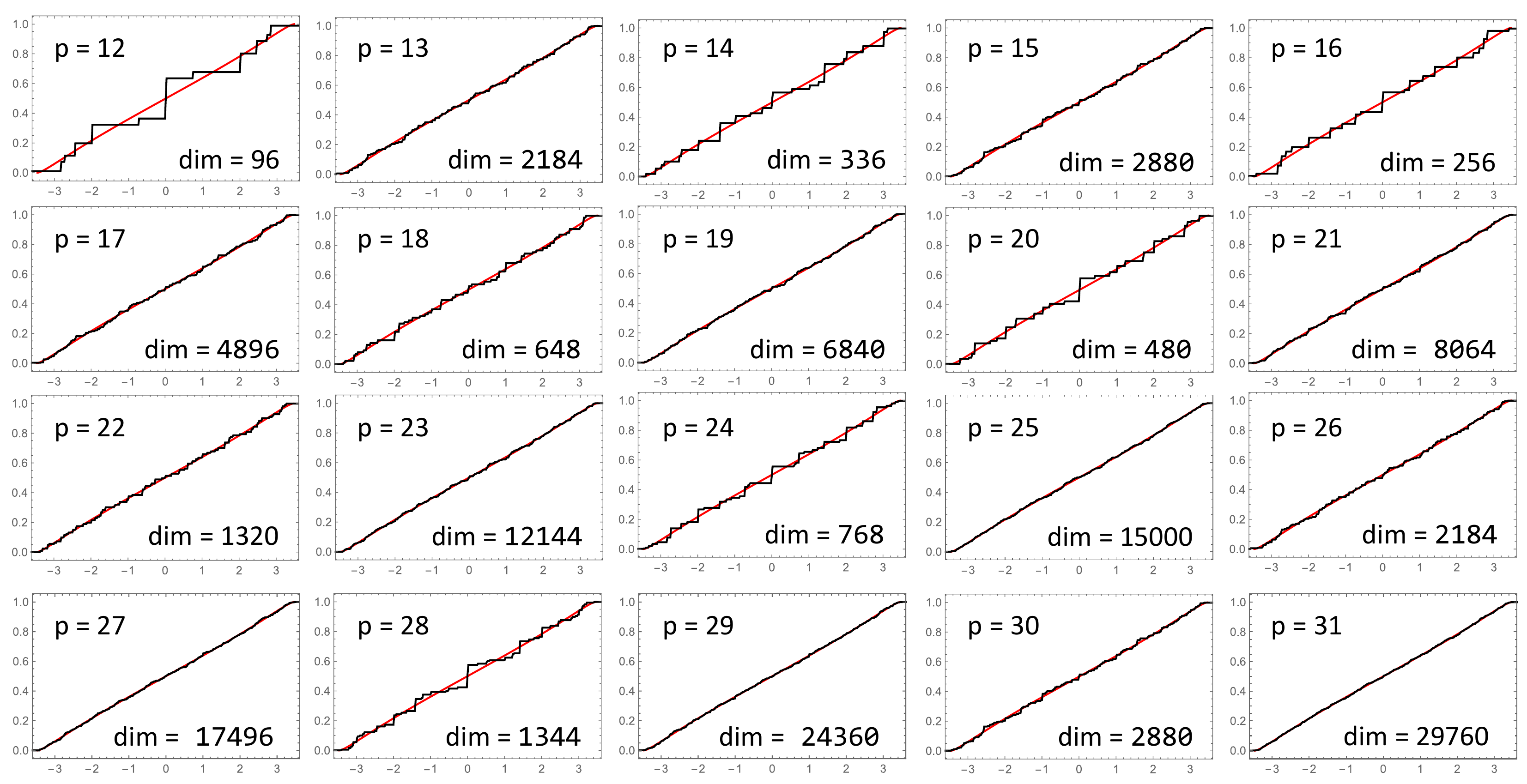}\\
  \caption{\small Spectral densities of the finite approximations $\pi_L(\phi_1(\delta))$ obtained with the converging periodic boundary conditions \eqref{Eq:ConvApp1} for different values of $p$, compared with the exact spectral density (red line) of the adjacency matrix. At the bottom of each panel, we show the dimensions of the quotient groups, which determine the dimensions of the Hilbert spaces of the finite approximations.
}
 \label{Fig:IDS3}
\end{figure}

The general principle that we are relying on is as follows. Suppose we are dealing with a subgroup of ${\rm GL}(n,R)$, where $R$ is a ring with a unit. Then, if $J \subset R$ is an ideal of $R$, then
\begin{equation}\label{Eq:TildeGL}
\widetilde{\rm GL}(n,J) := \big \{ [\delta_{ij}+a_{ij}] \in {\rm GL}(n,R), \ a_{ij} \in J \big \}
\end{equation}
is a normal subgroup of  ${\rm GL}(n,R)$. Indeed, the elements of $\widetilde{\rm GL}(n,J)$ are of the form $I + Q$, with $I$ the identity matrix and $Q$ from the ideal $M_n(J)$ of the ring of $n\times n$ matrices $M_n(R)$. Then, we can easily see that the set~\eqref{Eq:TildeGL} is stable against matrix multiplication and, furthermore, for any $M \in {\rm GL}(n,R)$,
\begin{equation}
M \cdot (I + Q) \cdot M^{-1} = I + M \cdot Q \cdot M^{-1},
\end{equation}
which clearly belongs to $\widetilde{\rm GL}(n,J)$. Also, the class in ${\rm GL}(n,R)/\widetilde{\rm GL}(n,J)$ of an element of ${\rm GL}(n,R)$ can be computed as 
\begin{equation}
[a_{ij}] \cdot \widetilde{\rm GL}(n,J) = [a_{ij} \, {\rm mod}\, J].
\end{equation}
For example, in the case of $n=2$, we have
\begin{equation}
\begin{pmatrix} a & b \\ c & d \end{pmatrix} \cdot \begin{pmatrix} 1 + J & J \\ J & 1 + J \end{pmatrix} = \begin{pmatrix} a + J & b + J \\ c + J & d +J \end{pmatrix} = \begin{pmatrix} a \, {\rm mod}\, J & b \, {\rm mod}\, J \\ c \, {\rm mod}\, J & d \, {\rm mod}\, J \end{pmatrix}.
\end{equation}

We now invoke the matrix presentation of $\FM_2$ supplied in Eq.~\eqref{Eq:F2Mat} and observe that, for $p \in \NM$, $p>1$, 
\begin{equation}\label{Eq:ConvApp1}
\FM_2=G_0 \triangleright G_1=\FM_2 \cap \widetilde{\rm GL}(2,p\ZM) \triangleright G_2=\FM_2 \cap \widetilde{\rm GL}(2,p^2\ZM) \triangleright \cdots
\end{equation}
is a coherent sequence of normal subgroups such that $\cap \, G_N = \{e\}$.  The quotient group $H_N=\FM_2/G_N$ can be generated by producing free products of the generating matrices~\eqref{Eq:F2Mat} and their inverses, applying ${\rm mod}\, p^N$ on all the entries of these free products and, finally, deleting the duplicates. The outcome of this process is a finite set $H_N$ of matrices. Generating the multiplication tables for $H_N$'s amounts to taking pairwise matrix products from $H_N$, taking the ${\rm mod}\, p^N$ on the result and identifying this final result with an existing element from the set $H_N$. Once this information is gathered on a computer, the left regular representations of the elements of $\CM H_N$ can be obtained by standard procedures, in particular, that of the adjacency operators $\phi_N(\delta)$. The algorithm we just outlined is fully documented in \cite{FabianGH}. Fig.~\ref{Fig:IDS1} reports the spectral densities of $\pi_L(\phi_N(\delta))$ for different $N$'s and $p=2$. One can clearly see a fast convergence of the numerical results to the exact spectral density, as the theory predicted. In fact, the mean squared deviations $MSD=0.0053, 0.00076,0.000096,0.000025$ for $N=3,4,5,6$, respectively, of the numerical estimates from the exact spectral density indicate that the convergence happens exponentially fast in $N$. Lastly, for completion, we report in Fig.~\ref{Fig:IDS3} the spectral densities of $\pi_L(\phi_1(\delta)$ for different values of $p$ in Eq.~\eqref{Eq:ConvApp1}. As opposed to the monotone trend seen in Fig.~\ref{Fig:IDS1}, the quality of the numerical outputs seen in Fig.~\ref{Fig:IDS3} fluctuates wildly with $p$. This is because the subgroups generated by this procedure do not form a coherent tower.

The simulations demonstrate that our proposed periodic boundary conditions lead to fast converging algorithms and, in the following, we demonstrate that in fact we can reach a point where we can use the spectral simulations to derive combinatorial informations of the Cayley graph. For example, the number of closed loops of length $n$, starting and ending at the origin, can be estimated as 
\begin{equation}\label{Eq:NeApp1}
N_e(n) = \int x^n {\rm d} F(x) \approx |H_N|^{-1} \sum_{\lambda \in {\rm Spec}(\pi_L(\phi_N(\delta)))} \lambda^n,
\end{equation}
where $|H_N|$ is the cardinal of $H_N$. The right side is simply the spectral function of $\phi_N(\delta)$, computed as explained at the beginning of subsection~\ref{Sec:Misleading}. The outputs of this equation for increasing values of $N$, with $p=2$, are reported and compared with the exact result from Eq.~\eqref{Eq:Ng1}  in Fig.~\ref{Fig:NrPathsFreeGroup}(a). According to these results, we can indeed determined $N_e(n)$ exactly for $n$ as high as $20$, if we work with $N \geq 6$.

\begin{figure}[t]
\center
\includegraphics[width=0.8\textwidth]{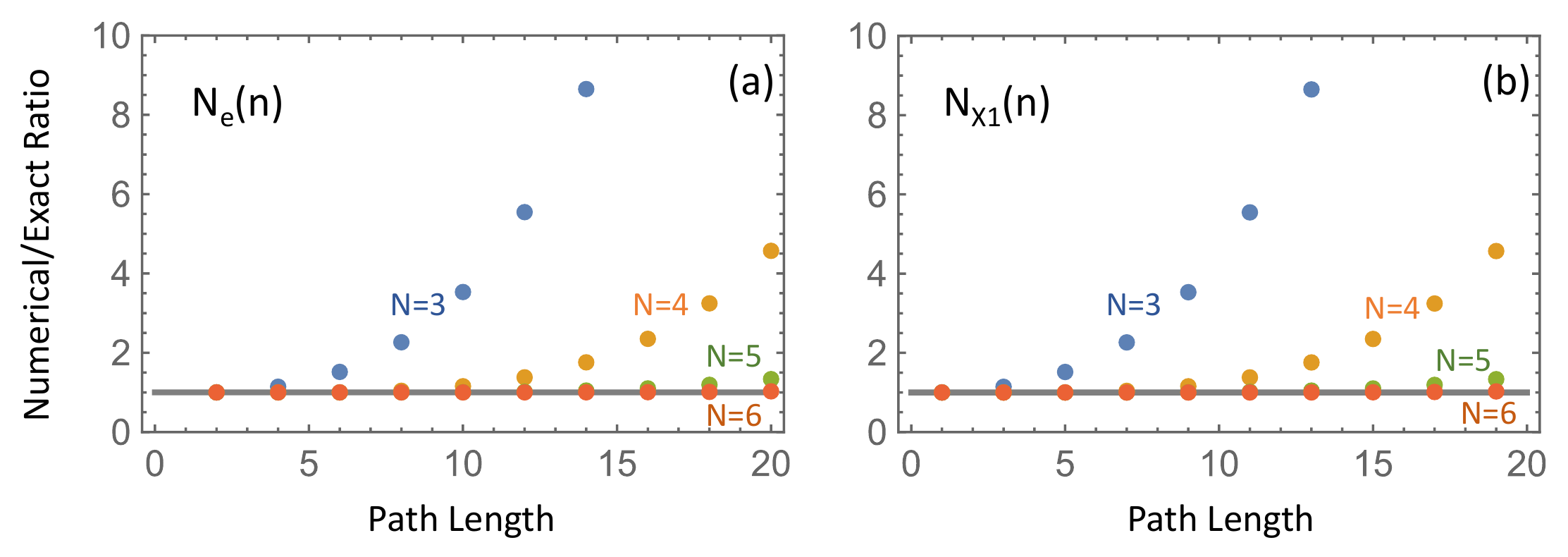}\\
  \caption{\small (a) The ratio between the output of Eq.~\eqref{Eq:NeApp1} and the exact value from Eq.~\eqref{Eq:Ng1} (gray line). (b) The ratio between the output of Eq.~\eqref{Eq:NeApp2} with $g=X_1$ and the exact value from Eq.~\eqref{Eq:Ng2} (gray line). The inputs for both numerical calculations consist of the spectral densities reported in Fig.~\ref{Fig:IDS1}. Eq.~\eqref{Eq:NgApp3} also leads to the same results, for both cases. 
}
 \label{Fig:NrPathsFreeGroup}
\end{figure}

Likewise, the number of paths starting at the node labeled by $g$ and ending at the node labeled by $e$ can be estimated as follows, 
\begin{equation}\label{Eq:NeApp2}
N_g(n) = \int  x^n \, {\rm d} \langle e |\bm E(x) |g \rangle \approx \sum_{\lambda \in {\rm Spec}(\pi_L(\phi_N(\delta)))} \lambda^n \,  \psi_\lambda(e)^\ast \psi_\lambda(\phi_N(g)),
\end{equation}
where $\psi_\lambda$ is the eigenvector of $\phi_N(\delta)$ corresponding to the eigenvalue $\lambda$. Of course, the above expression reduces to the one from Eq.~\eqref{Eq:NeApp1} when $g = e$. Additionally and more straightforwardly, $N_g(n)$ can be estimated using Eq.~\eqref{Eq:Ng0} and its approximations:
\begin{equation}\label{Eq:NgApp3}
N_g(n) = \langle e | \pi_L(\delta)^n |g\rangle \approx \langle e | \pi_L(\phi_N(\delta))^n |\phi_N(g)\rangle
\end{equation}
In Fig.~\ref{Fig:NrPathsFreeGroup}(b), the output of Eq.~\ref{Eq:NeApp2} to the input reported in Fig.~\ref{Fig:IDS1} and its corresponding egenvectors is compared with the exact result from Eq.~\eqref{Eq:Ng2}. Eq.~\eqref{Eq:NgApp3} leads to exactly the same results. The matching between the two, seen there, is a validation that the algorithm, indeed, correctly reproduces the off-diagonal matrix elements of the resolvent.

\begin{figure}[t]
\center
\includegraphics[width=0.8\textwidth]{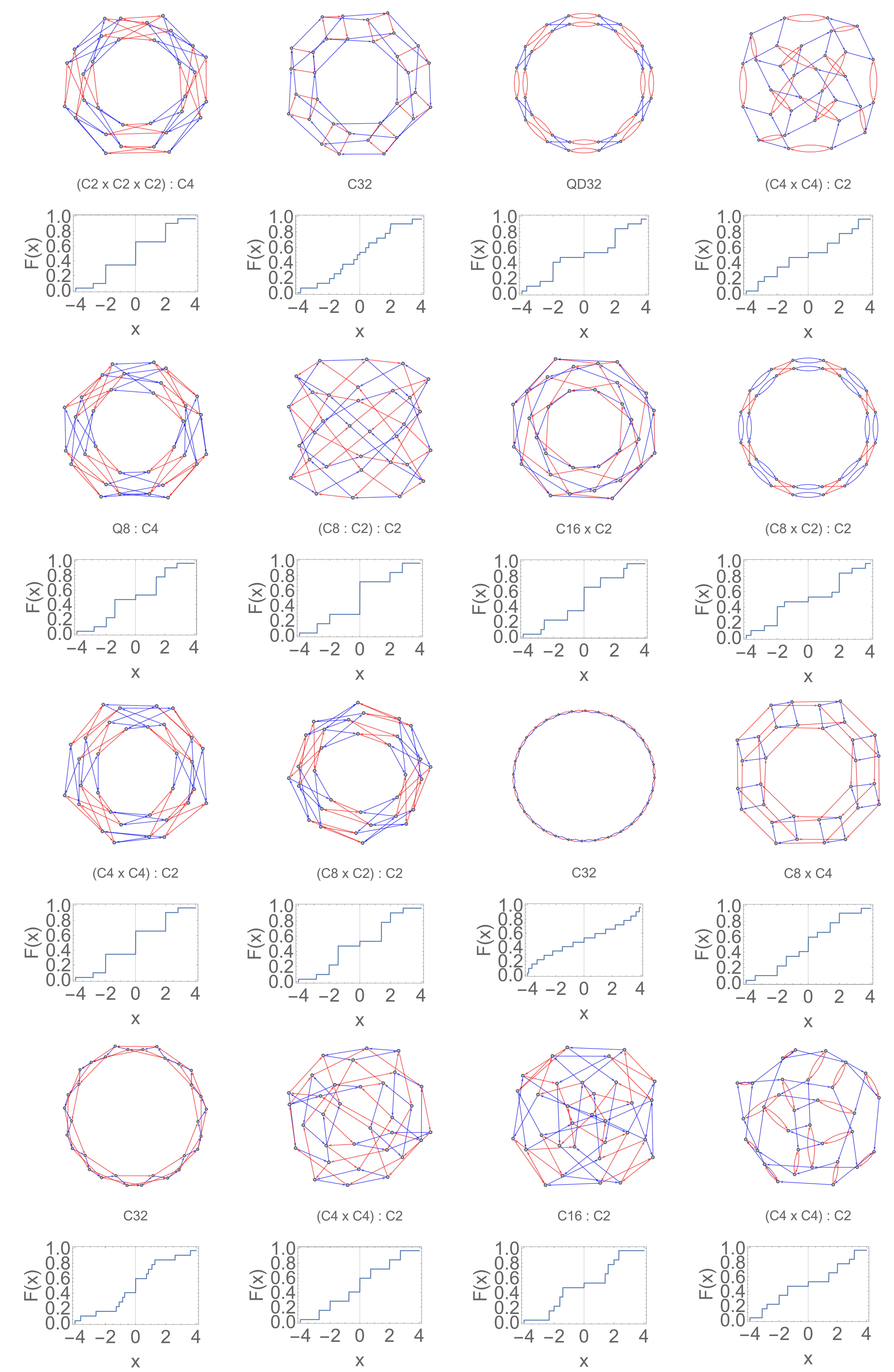}\\
  \caption{\small A kaleidoscope of finite size approximations generated with the software \cite{RoberLINS}.
}
 \label{Fig:FGroups}
\end{figure}

Lastly, we followed \cite{MaciejkoPNAS2022} in using the existing algebra software \cite{RoberLINS}, to produce normal subgroups of $\FM_2$. Our initial interest for doing so was to explore if there are normal subgroups out there that supply a faster convergence than the ones generated by our particular subgroups. So far we have not found any. In fact, our coherent sequences of normal subgroups by far outperform the many other subgroups we explored with the code \cite{RoberLINS}. In Fig.~\ref{Fig:FGroups}, we present a kaleidoscope of small finite Cayley graphs resulted from this exercise, together with the associated spectral densities of the adjacency operator. Let us make clear that the normal groups with larger indices seen in Fig.~\ref{Fig:IDS3} are completely out of the range of \cite{RoberLINS}. Another observation that is relevant for experiments is that these graphs quickly become very complicated when increasing the size of the approximation. This is, unfortunately inevitable because folding a truncated graph into itself produces many connections that cross each other. For this reason, the only way to communicate the finite periodic approximations to the experimental laboratories is through a table akin to a multiplication table of a group, this time containing the strength of the coupling coefficients of the reduced Hamiltonian.

\subsection{Converging finite approximations for a Fuchsian group}\label{Sec:ConvFf2}

\begin{figure}[t]
\center
\includegraphics[width=\linewidth]{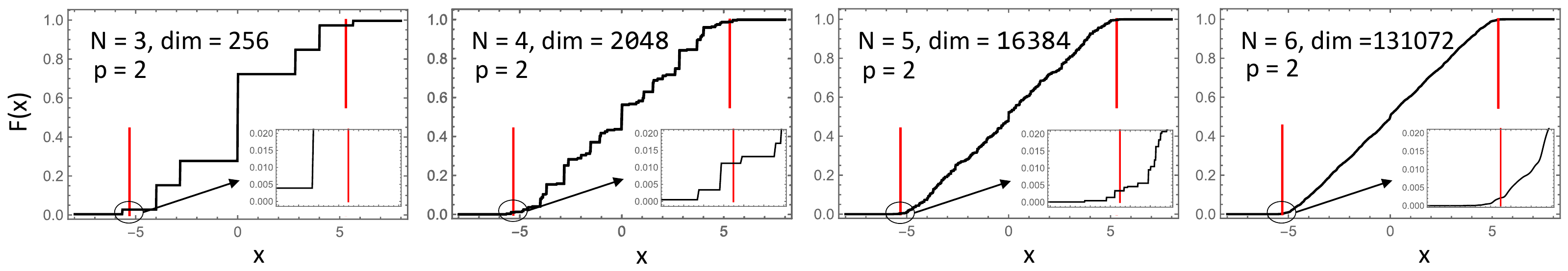}\\
  \caption{\small Spectral densities of the finite approximations $\pi_L(\phi_N(\delta))$ obtained with the converging periodic boundary conditions \eqref{Eq:ConvApp2} with $p=2$. Next to the values of $N$, we show the dimensions of the quotient groups, which determine the dimensions of the Hilbert spaces of the finite approximations. The vertical red bars indicate the existing rigorous estimates of the edges of the spectrum of $\Delta$ on the standard Cayley graph of $\Ff_2$.
}
 \label{Fig:IDS1F}
\end{figure}

We start with the matrix presentation of the Fuchsian group $\Ff_2$ supplied in Eq.~\eqref{Eq:Ff2Mat} and observe again that, for $p \in \NM$, $p>1$, 
\begin{equation}\label{Eq:ConvApp2}
\Ff_2=G_0 \triangleright G_1=\Ff_2 \cap \widetilde{\rm GL}(2,p\ZM+\sqrt{3}\, p\ZM) \triangleright G_2=\Ff_2 \cap \widetilde{\rm GL}(2,p^2\ZM+\sqrt{3}\, p^2\ZM) \triangleright \cdots
\end{equation}
is a coherent sequence of normal subgroups such that $\cap \, G_N = \{e\}$.  The quotient groups $H_N=\Ff_2/G_N$ and their multiplication tables, as well as the left regular representations, can be generated via the procedures very similar to the ones already described in the previous subsection. Hence, we will jump directly to the numerical results. Fig.~\ref{Fig:IDS1F} reports the spectral densities of $\pi_L(\phi_N(\delta))$ for different $N$'s and $p=2$. Unfortunately, we are not aware of combinatorial results similar to what was presented for the free groups and, as such, we cannot directly compare the numerical results against an exact ones. However, we do have indirect checks we can perform. First, it is known that the spectral radius $\mu$ of the adjacency operator satisfies the rigorous bounds \cite{Bartholdi2004,GouezelCPC2015} (see also \cite{KestenTAMS1959,KestenMS1959,PaschkeMZ1993,ZukCM1997,BartholdiCM1997}):
\begin{equation}
8 \times 0.6627 \leq \mu \leq 8 \times 0.6629.
\end{equation}
Thus, we know for sure that the spectrum of $\pi_L(\delta)$ is contained in the interval $[-5.3032,5.3032]$ and contains the interval $[-5.3016,5.3016]$. These are very similar intervals and, in Fig.~\ref{Fig:IDS1F}, they are marked by the red vertical lines. Our numerically computed spectra seem to rapidly converge to this mentioned interval. 

Secondly, it is not difficult to see that the counting of closed walks from Eq.~\eqref{Eq:Ng1} still apply for paths of length less than eight, which explains why the spectral densities in Fig.~\ref{Fig:IDS1F} are similar to the ones reported in Fig.~\ref{Fig:IDS1}. Using that formula, we find $N_e(2)=8$, $N_e(4)=120$ and $N_e(6)=2192$. We also predict $N_e(8)=44264$, a value that equals the output of Eq.~\eqref{Eq:Ng1} plus 16 additional loops (the eight elementary close loops seen in Fig.~\ref{Fig:CayleyGraphs}(b) walked in two different directions). On the other hand, the outputs of Eq.~\eqref{Eq:NeApp1} to the inputs from Fig.~\ref{Fig:IDS1F}, or equivalently the outputs of Eq.~\eqref{Eq:NgApp3}, are 
\begin{equation*}
 \begin{matrix} \ & n=2 & n=4 & n=6 & n=8 &  \ \\
\ & \downarrow &\downarrow & \downarrow & \downarrow &  \ \\
N_e(n): & 8 & 160 & 4736 & 197632  & (N=3)\\
N_e(n): & 8 & 120 & 2384 & 61680  & (N=4) \\
N_e(n): & 8 & 120 & 2192 & 45544  & (N=5) \\
N_e(n): & 8 & 120 & 2192 & 44536  & (N=6) \\
N_e(n): & 8 & 120 & 2192 & 44296  & (N=7)
\end{matrix}
\end{equation*}
These results reproduce $N_e(n)$ exactly, for $n \leq 6$, and also $N_e(8)$ with up to 0.072\% error. We are confident that $N_e(8)$ will eventually converge to its exact value.\footnote{The dimension of the Hilbert space for $N=8$ finite approximation is 1,048,576 and going beyond $N=7$ will require significant computational resources.}

To conclude, our numerical simulations display a rapid convergence and they pass a number of quantitative tests, which emboldens us to assert that the data reported in Fig.~\ref{Fig:IDS1F} is the most accurate representation to date of the exact spectral density of the adjacency operator on the Cayley graph of $\Ff_2$.

\end{document}